\documentclass[a4paper, english, runningheads]{article}

\usepackage[utf8]{inputenc}
\usepackage{amsmath}
\usepackage{amsfonts}
\usepackage{amssymb}
\usepackage{amsthm}
\usepackage{xcolor,graphicx}
\usepackage{tabularx}
\usepackage{booktabs}
\usepackage{pdflscape}
\usepackage{xspace}
\usepackage{dsfont}

\usepackage{rotating}
\usepackage[numbers,sort,sectionbib]{natbib}
\usepackage[vlined,linesnumbered,ruled]{algorithm2e}
\SetEndCharOfAlgoLine{}
\SetCommentSty{textrm}
\SetKwProg{Fn}{Function}{}{}

\usepackage[rightcaption]{sidecap}
\sidecaptionvpos{figure}{c}

\usepackage{tikz}
\usetikzlibrary{matrix,arrows,decorations.pathmorphing,decorations.pathreplacing,backgrounds,positioning,fit,shapes,arrows,calc,patterns,snakes}
\tikzstyle{edge} = [color=black,opacity=.15,line cap=round, line join=round, line width=8pt]
\usepackage{authblk}

\usepackage{scrpage2}
\usepackage{a4wide}

\usepackage{todonotes}

\usepackage[
	pdfdisplaydoctitle,
	menucolor=orange!40!black,
	filecolor=magenta!40!black,
	urlcolor=blue!40!black,
	linkcolor=red!40!black,
	citecolor=green!40!black,
	pdftitle={The Power of Data Reduction for Matching}, pdfauthor={G. B. Mertzios, A. Nichterlein, R. Niedermeier},
	colorlinks,pagebackref
	]{hyperref}

\usepackage[nameinlink,sort&compress,capitalize]{cleveref}

\theoremstyle{plain}
\newtheorem{theorem}{Theorem}[section]
\newtheorem{lemma}[theorem]{Lemma}
\newtheorem{corollary}[theorem]{Corollary}
\newtheorem{proposition}[theorem]{Proposition}
\newenvironment{proofofclaim}[1][Proof]{\noindent\textbf{#1.} }{\hfill\ensuremath{\square}}
\newtheorem{myclaim}[theorem]{Claim}
\newtheorem{observation}[theorem]{Observation}
\newtheorem{rrule}{Reduction Rule}[section]

\crefname{rrule}{Reduction Rule}{Reduction Rules}
\crefname{cond}{Condition}{Conditions}
\crefname{myclaim}{Claim}{Claims}
\crefname{step}{Step}{Steps}

\theoremstyle{definition}
\newtheorem{definition}{Definition}

\theoremstyle{remark}

\theoremstyle{plain}

\newcommand{\problemdef}[3]{
	\begin{center}
		\begin{minipage}{0.95\textwidth}
			\textsc{#1}
			
			\vspace{3pt}
			
			\setlength{\tabcolsep}{3pt}
			\begin{tabularx}{\textwidth}{@{}lX@{}}
				\textbf{Input:} 	& #2 \\
				\textbf{Question:} 	& #3
			\end{tabularx}
		\end{minipage}
	\end{center}
}

\newcommand{\Match}{\textsc{Matching}\xspace}

\newcommand{\BipMatch}{\textsc{Bipartite Matching}\xspace}
\newcommand{\paramEnv}[1]{``{#1}''}

\DeclareMathOperator{\lmv}{\#lmv}
\DeclareMathOperator{\rmv}{\#rmv}

\DeclareMathOperator{\noAug}{\#aug}
\DeclareMathOperator{\tab}{Tab}

\newcommand{\distPara}{k}
\newcommand{\solSize}{s}

\title{The Power of Linear-Time Data Reduction \\ for Maximum Matching}

\author[1]{George B. Mertzios}
\author[1,2]{André Nichterlein\thanks{Supported by a postdoc fellowship of the German Academic Exchange Service (DAAD) while at Durham University.}}
\author[2]{Rolf Niedermeier}

\affil[1]{School of Engineering and Computing Sciences, Durham University, UK, 
\texttt{george.mertzios@durham.ac.uk}}
\affil[2]{Institut f\"ur Softwaretechnik und Theoretische Informatik,  TU Berlin, Germany,
 \texttt{\{andre.nichterlein,rolf.niedermeier\}@tu-berlin.de}}
\date{}

\begin{document}

\maketitle

\thispagestyle{scrheadings}
\cfoot{}
\ohead{}
\ifoot{}

\begin{abstract}
Finding maximum-cardinality matchings in undirected graphs is arguably one of the most central 
graph primitives. 
For $m$-edge and $n$-vertex graphs, it is well-known to be solvable in $O(m\sqrt{n})$~time; however, for several applications this running time is still too slow. 
We investigate how linear-time (and almost linear-time) data reduction (used 
as preprocessing) can alleviate the situation.
More specifically, we focus on \emph{linear-time kernelization}. 
We start a deeper and systematic study
both for general graphs and for bipartite graphs. Our data reduction algorithms easily 
comply (in form of preprocessing) with every solution strategy (exact, approximate, heuristic), thus making them attractive in various settings.
\end{abstract}

\section{Introduction}

``Matching is a powerful piece of algorithmic magic''~\cite{Ski10}. 
In the \textsc{Maximum Matching} problem, given an undirected graph, one has to compute a maximum-cardinality set of nonoverlapping edges. 
Maximum matching is arguably among the most fundamental graph-algorithmic primitives allowing for a polynomial-time algorithm. 
More specifically, on an $n$-vertex and $m$-edge graph a maximum matching can be found in $O(m\sqrt{n})$ time~\cite{MV80}. 
Improving this upper time bound %
resisted decades of research. 
Recently, however, \citet{DP14} presented a linear-time algorithm that computes a~$(1-\epsilon )$-approximate maximum-weight matching, 
where the running time dependency on $\epsilon$ is $\epsilon^{-1}\log(\epsilon^{-1})$.
For the unweighted case, the~$O(m \sqrt{n})$ algorithm of \citet{MV80} implies a linear-time~$(1-\epsilon)$-approximation, where in this case the running time dependency on~$\epsilon$ is~$\epsilon^{-1}$~\cite{DP14}.
We take a different route: 
First, we do not give up the quest for optimal solutions.
Second, we focus on efficient---more specifically, linear-time executable---data reduction rules, that is, not solving an instance but significantly shrinking its size before actually solving the problem.\footnote{Doing so, however, we focus on the unweighted case.}
In the context of decision problems and parameterized algorithmics this approach is known as kernelization; this is a particularly active area of algorithmic research on NP-hard problems.

The spirit behind our approach is thus closer to the identification of efficiently solvable special cases of \textsc{Maximum Matching}. 
There is quite some body of work in this direction.
For instance, since an augmenting path can be found in linear time~\cite{GT85}, the standard augmenting path-based algorithm runs in $O(s(n+m))$~time, where~$s$ is the number of edges in the maximum matching.
\citet{Yus13} developed an $O(rn^2\log n)$-time algorithm, where $r$ is the difference between maximum and minimum degree of the input graph. 
Moreover, there are linear-time algorithms for computing maximum matchings in special graph classes, including convex bipartite~\cite{SY96}, strongly chordal~\cite{DK98}, chordal bipartite graphs~\cite{Cha96}, and cocomparability graphs~\cite{MNN17}. 

All this and the more general spirit of ``parameterization for polynomial-time solvable problems'' (also referred to as ``FPT~in~P'' or ``FPTP'' for short)~\cite{GMN17} forms the starting point of our research.  
Remarkably, \citet{FLPSW18} recently developed an algorithm to compute a maximum matching in graphs of treewidth~$k$ in $O(k^4 n\log^2 n)$ randomized time. 
Afterwards, \citet{IOO18} provided an elegant algorithm computing a maximum matching in graphs of treedepth~$\ell$ in $O(\ell \cdot m)$ time. 
This implies an~$O(k^2 n \log n)$-time algorithm where~$k$ is the treewidth, since $m \in O(k n)$ and $\ell \le (k+1)\log n$~\cite{NO12}.
\citet{CDP19} provided an $O(r^4 n + m)$-time algorithm where~$r$ is the modular-width.
\citet{KN18} improved this to $O(r^2 \log(r) n + m)$ time.

Following the paradigm of \emph{kernelization}, that is, provably effective and efficient data reduction, we provide a systematic exploration of the power of not only polynomial-time but actually linear-time data reduction for \textsc{Maximum Matching}. 
Thus, our aim (fitting within FPTP) is to devise problem kernels that are computable in linear time. 
In other words, the fundamental question we pose is whether there is a very efficient preprocessing that provably shrinks the input instance, where the effectiveness is measured by employing some parameters.
The philosophy behind is that if we can design linear-time data reduction algorithms, then we may employ them for free before afterwards employing any super-linear-time solving algorithm.
We believe that this sort of question deserves deeper investigation and we initiate it based on the \textsc{Maximum Matching} problem.
In fact, in follow-up work we demonstrated that such linear-time data reduction rules can significantly speed-up state-of-the-art solvers for \Match~\cite{KNNZ18}.

As kernelization is usually defined for decision problems, we use in the remainder of the paper the \emph{decision version} of \textsc{Maximum Matching}. 
In the rest of the paper we call this decision version \Match. 
In a nutshell, a kernelization of a decision problem instance is an algorithm that produces an equivalent instance whose size can solely be upper-bounded by a function in the parameter (preferably a polynomial function).
The focus on decision problems is justified by the fact that all our results, although formulated for the decision version, in a straightforward way extend to the corresponding optimization version (as also done in our follow-up work~\cite{KNNZ18}).
\problemdef{(\textsc{Maximum-Cardinality}) \Match}
	{An undirected graph~$G=(V,E)$ and a nonnegative integer~$s$.}
	{Is there a size-$s$ subset~$M_G \subseteq E$ of nonoverlapping (i.e.~disjoint) edges?}
Note that, for any polynomial-time solvable problem, solving the given instance and returning a trivial yes- or no-instance always produces a constant-size kernel in polynomial time.
Hence, we are looking for kernelization algorithms that are \emph{faster} than the algorithms solving the problem.
The best we can hope for is linear time.
For NP-hard problems, each polynomial-time kernelization algorithm is faster than any solution algorithm, unless P=NP.
While the focus of classical kernelization for NP-hard problems is mostly on improving the \emph{size} of the kernel, we particularly emphasize that for polynomially solvable problems it is mandatory to also focus on the \emph{running time} of the kernelization algorithm. 
Indeed, we can consider linear-time kernelization as the holy grail and this drives our research when studying kernelization for \Match.

\paragraph{Our contributions.}
We present three kernels for \Match (see \cref{tab:kernel-results} for an overview). %
\begin{table}
	\caption{
		Our kernelization results.
	}
	\label{tab:kernel-results}
	\begin{tabularx}{\textwidth}{lXlX}
		\toprule
		Parameter~$k$					& running time 	& kernel size & \\
		\midrule
		\multicolumn{4}{l}{\textbf{Results for \Match}} \\
		Feedback edge number		& $O(n+m)$ time & $O(k)$ vertices and edges & (\cref{thm:fes-lin-kernel}) \\
		Feedback vertex number		& $O(kn)$ time  & $2^{O(k)}$ vertices and edges & (\cref{thm:fvs-kernel}) \\
		\midrule
		\multicolumn{4}{l}{\textbf{Results for \BipMatch}} \\
		Distance to chain graphs 	& $O(n+m)$ time & $O(k^3)$ vertices & (\cref{thm:cubic-kernel-chain-graphs}) \\
		\bottomrule
	\end{tabularx}
\end{table}
All our parameterizations can be categorized as ``distance to triviality''~\cite{Cai03,GHN04}. 
They are motivated as follows. 
First, note that it is important that the parameters we exploit can be computed, or well approximated (within constant factors), in linear time regardless of the parameter value.
Next, note that maximum-cardinality matchings can be trivially found in linear time on trees (or forests). 
That is why we consider the edge deletion distance (feedback edge number) and vertex deletion distance (feedback vertex number) to forests.
Notably, there is a trivial linear-time algorithm for computing the feedback edge number and there is a linear-time factor-4 approximation algorithm for the feedback vertex number~\cite{BGNR98}.
We mention in passing that the parameter vertex cover number, which is lower-bounded by the feedback vertex number, has been frequently studied for kernelization. %
In particular, \citet{GP13,GMN17} provided a linear-time computable quadratic-size kernel for \Match with respect to the parameter solution size (or equivalently vertex cover number). 
Coming to bipartite graphs, we parameterize by the vertex deletion distance to chain graphs which is motivated as follows. 
First, chain graphs form one of the most obvious easy cases for bipartite graphs where 
\Match can be solved in linear time~\cite{SY96}. 
Second, we show that the vertex deletion distance of any bipartite graph to a chain graph 
can be 4-approximated in linear time. %
Moreover, vertex deletion distance to chain graphs lower-bounds the vertex cover number of a bipartite graph.
 
An overview of our main results is given in \cref{tab:kernel-results}.
We study kernelization for \Match parameterized by the feedback vertex number, 
that is, the vertex deletion distance to a forest (see \cref{sec:general-case}). 
As a warm-up we first show that a subset of our data reduction rules for the ``feedback vertex set kernel'' also yields a linear-time computable linear-size kernel for the typically much larger parameter feedback edge number (see \cref{sec:fes-kernel}). 
As for \BipMatch no faster algorithm is known than on general graphs, 
we kernelize \BipMatch with respect to the vertex deletion distance to chain graphs (see \cref{sec:bipartite-case}).

Seen from a high level, our two main results (\cref{thm:fvs-kernel,thm:cubic-kernel-chain-graphs}, see \cref{tab:kernel-results}) employ the same 
algorithmic strategy, namely upper-bounding 
(as a function of the parameter) the number of neighbors in the appropriate vertex deletion set (also called modulator)~$X$; 
that is, in the feedback vertex set or in the deletion set to chain graphs, respectively. 
To achieve this we develop new \emph{``irrelevant edge techniques''} tailored to these two kernelization problems. 
More specifically, whenever a vertex~$v$ of the deletion set~$X$ has large degree, then we efficiently detect edges incident to $v$ whose removal does not change the size of the maximum matching. 
Then the remaining graph can be further shrunk by scenario-specific data reduction rules. 
While this approach of removing irrelevant edges is natural, the technical details and the proofs of correctness become quite technical and combinatorially challenging. 
Note that there exists a trivial~$O(km)$-time solving (not only kernelization) algorithm, where~$k$ is the feedback vertex number. 
Our kernel has size $2^{O(k)}$. Therefore, only if~$k = o(\log n)$ our kernelization algorithm provably shrinks the initial instance. 
However, our result is still relevant: 
First, our data reduction rules might assist in proving a polynomial upper bound---so our result is a first step in this direction. 
Second, the running time~$O(kn)$ of our kernelization algorithm is a kind of ``half way'' between~$O(km)$ (which could be as bad as~$O(k^2n)$) and~$O(n+m)$ (which is best possible).
Finally, note that this work focuses on theoretical and worst-case analysis; in practice, our kernelization algorithm might achieve much better upper bounds on real-world input instances.
In fact, in experiments using the kernelization with respect to the feedback edge number, the observed kernels were always significantly smaller than the theoretical bound~\cite{KNNZ18}.

As a technical side remark, we emphasize that in order to achieve a linear-time kernelization algorithm, we often need to use suitable data structures and to 
carefully design the appropriate data reduction rules to be exhaustively applicable in linear time, making this form of ``algorithm engineering'' much more 
relevant than in the classical setting of mere polynomial-time data reduction rules.

\section{Preliminaries and basic observations} \label{sec:prelim}

\paragraph{Notation and Observations.} 
We use standard notation from graph theory. 
A \emph{feedback vertex (edge) set} of a graph~$G$ is a set~$X$ of vertices (edges) such that~$G-X$ is a tree or forest.
The \emph{feedback vertex (edge) number} denotes the size of a minimum feedback vertex (edge) set.
All paths we consider are simple paths. 
Two paths in a graph are called \emph{internally vertex-disjoint} if they are either completely vertex-disjoint or they overlap only in their endpoints. 
A \emph{matching} in a graph is a set of pairwise disjoint edges.
Let~$G = (V,E)$ be a graph and let~$M \subseteq E$ be a matching in~$G$.
The degree of a vertex is denoted by~$\deg(v)$.
A vertex~$v \in V$ is called \emph{matched} with respect to~$M$ if there is an edge in~$M$ containing~$v$, otherwise~$v$ is called \emph{free} with respect to~$M$.
If the matching~$M$ is clear from the context, then we omit ``with respect to~$M$''.
An \emph{alternating path} with respect to~$M$ is a path in~$G$ such that every second edge of the path is in~$M$.
An \emph{augmenting path} is an alternating path whose endpoints are free.
It is well known that a matching~$M$ is maximum if and only if there is no augmenting path for it. %
Let~$M \subseteq E$ and~$M' \subseteq E$ be two matchings in~$G$. 
We denote by~$G(M,M') := (V, M \bigtriangleup M')$ the graph containing only the edges in the symmetric difference of~$M$ and~$M'$, that is, $M \bigtriangleup M' := M \cup M' \setminus (M \cap M')$.
Observe that every vertex in~$G(M,M')$ has degree at most two.

For a matching~$M \subseteq E$ for~$G$ we denote by~$M^{\max}_G(M)$ a maximum matching in~$G$ with the largest possible overlap (in number of edges) with~$M$. That is, $M^{\max}_G(M)$ is a maximum matching in~$G$ such that for each maximum matching~$M'$ for~$G$ it holds that ${|M \bigtriangleup M'| \ge |M \bigtriangleup M^{\max}_G(M)|}$.
Observe that if~$M$ is a maximum matching for~$G$, then~$M^{\max}_G(M) = M$.
Furthermore observe that~$G(M,M^{\max}_G(M))$ consists of only odd-length paths and isolated vertices, and each of these paths is an augmenting path for~$M$. 
Moreover the paths in~${G(M, \allowbreak M^{\max}_G(M))}$ are as short as possible:
\begin{observation}\label{obs:shortest-augmenting-paths}
	For any path~$v_1, v_2, \ldots, v_p$ in~$G(M,M^{\max}_G(M))$ it holds that~$\{v_{2i-1},v_{2j}\} \notin E$ for every~$1 \le i < j \le p/2$.
\end{observation}

\begin{proof}
	Assume that~$\{v_{2i-1},v_{2j}\} \in E$.
	Then~$v_1, v_2, \ldots, v_{2i-2},v_{2i-1},v_{2j},v_{2j+1}, \allowbreak\ldots, v_p$ is a shorter path which is also an augmenting path for~$M$ in~$G$.
	The corresponding maximum matching~$M'$ satisfies~$|M \bigtriangleup M^{\max}_G(M)| > |M \bigtriangleup M'|$, a contradiction to the definition of~$M^{\max}_G(M)$. \qed
\end{proof}

It is easy to see that removing~$k$ vertices in a graph can reduce the maximum matching size by at most~$k$:

\begin{observation}\label{obs:MatchingSizeGandG-X}
	Let~$G = (V,E)$ be a graph with a maximum matching~$M_G$, let~$X \subseteq V$ be a vertex subset of size~$k$, and let~$M_{G-X}$ be a maximum matching for~$G-X$.
	Then, $|M_{G-X}| \le |M_G| \le |M_{G-X}| + \distPara$.
\end{observation}

\paragraph{Kernelization.}
A \emph{parameterized problem} is a set of instances~$(I,k)$ where~$I \in\Sigma^*$ for a finite alphabet $\Sigma$, and~$k\in \mathbb{N}$ is the \emph{parameter}.
We say that two instances~$(I,k)$ and $(I',k')$ of parameterized problems~$P$ and~$P'$ are \emph{equivalent} if $(I,k)$ is a yes-instance for~$P$ if and only if $(I',k')$ is a yes-instance for~$P'$. 
A \emph{kernelization} is an algorithm that, given an instance~$(I,k)$ of a parameterized problem~$P$, computes in polynomial time an equivalent instance~$(I',k')$ of~$P$ (the \emph{kernel}) such that $|I'|+k'\leq f(k)$ for some  computable function~$f$. %
We say that~$f$ measures the \emph{size} of the kernel, and if~$f(k)\in k^{O(1)}$, then we say that $P$ admits a polynomial kernel. 
Typically, a kernel is achieved by applying polynomial-time executable data reduction rules.
We call a data reduction rule~$\mathcal{R}$ \emph{correct} if the new instance~$(I',k')$ that results from applying~$\mathcal{R}$ to~$(I,k)$ is equivalent to~$(I,k)$.
An instance is called \emph{reduced} with respect to some data reduction rule if further application of this rule has no effect on the instance.

\section{Kernelization for Matching on General Graphs} \label{sec:general-case}

In this section, we investigate the possibility of efficient and effective preprocessing for \Match.
As a warm-up, we first present in \cref{sec:fes-kernel} a simple, linear-size kernel for \Match with respect to the parameter \paramEnv{feedback edge number}.
Exploiting the data reduction rules and ideas used for this kernel, we then present in \cref{sec:fvs-kernel} the main result of this section: an exponential-size kernel for the almost always significantly smaller parameter \paramEnv{feedback vertex number}.

\subsection{Warm-up: Parameter feedback edge number} \label{sec:fes-kernel}
We provide a linear-time computable linear-size kernel for \Match parameterized by the feedback edge number, that is, the size of a minimum feedback edge set.
Observe that a minimum feedback edge set can be computed in linear time via a simple depth-first search or breadth-first search.
The kernel is based on the next two simple data reduction rules due to \citet{KS81}. 
These rules deal with vertices of degree at most two.

\begin{rrule}\label{rule:deg-zero-one-vertices}
	Let~$v \in V$.
	If~$\deg(v) = 0$, then delete~$v$.
	If~$\deg(v) = 1$, then delete~$v$ and its neighbor and decrease the solution size~$\solSize$ by one ($v$ is matched with its neighbor).
\end{rrule}

\begin{rrule}\label{rule:deg-two-vertices}
	Let~$v$ be a vertex of degree two and let~$u,w$ be its neighbors.
	Then remove~$v$, merge~$u$ and~$w$, and decrease the solution size~$\solSize$ by one.
\end{rrule}

The correctness was stated by \citet{KS81}. For the sake of completeness, we give a proof.

\begin{lemma}\label[lemma]{lem:rule-deg-0-1-2-correct}
	\cref{rule:deg-zero-one-vertices,rule:deg-two-vertices} are correct.
\end{lemma} 

\begin{proof}
	If~$v$ has degree zero, then clearly~$v$ cannot be in any matching and we can remove~$v$.
	
	If~$v$ has degree one, then let~$u$ be its single neighbor. 
	Let~$M$ be a maximum matching of size at least~$\solSize$ for~$G$. 
	Then~$u$ is matched in~$M$ since otherwise adding the edge~$\{u,v\}$ would increase the size of the matching.
	Thus, a maximum matching in~$G' = G - u -v$ has size at least~$\solSize-1$.
	Conversely, a maximum matching of size~$\solSize-1$ in~$G'$ can easily be extended by the edge~$\{u,v\}$ to a maximum matching of size~$\solSize$ in~$G$.
	
	If~$v$ has degree two, then let~$u$ and~$w$ be its two neighbors.
	Let~$M$ be a maximum matching of size at least~$\solSize$. 
	If~$v$ is matched in~$M$ (i.e.~either with the edge $\{u,v\}$ or with the edge $\{v,w\}$), then deleting~$v$ and merging~$u$ with~$w$ decreases the size of~$M$ by one. 
	Similarly, if~$v$ is not matched in~$M$, then both $u$ and~$w$ are matched in~$M$, since otherwise adding the edge~$\{u,v\}$ (resp.~$\{v,w\}$) would increase the size of the matching, a contradiction.
	Thus, in this case, deleting~$v$ and merging~$u$ with~$w$ decreases again the size of~$M$ by one ($M$~looses either the edge incident to~$v$ or one of the edges incident to~$u$ and~$w$).
	Hence, the resulting graph~$G''$ has a maximum matching of size at least~$\solSize - 1$.
	Conversely, let~$M''$ be a matching of size at least~$\solSize - 1$ for~$G''$.
	If the merged vertex~$vw$ is free, then~$M := M'' \cup \{\{u,v\}\}$ is a matching of size~$\solSize$ in~$G$.
	Otherwise, $vw$ is matched to some vertex~$y$ in~$M''$.
	Then matching~$y$ in~$G$ with either~$v$ or~$w$ (at least one of the two vertices is a neighbor of~$y$) and matching~$u$ with the other vertex yields a matching of size at least~$\solSize$ for~$G$.  \qed
\end{proof}

While it is easy to exhaustively apply \cref{rule:deg-zero-one-vertices} in linear time, applying \cref{rule:deg-two-vertices} exhaustively in linear time is nontrivial~\cite{BK09a}.
Note that applying \cref{rule:deg-two-vertices} might create new degree-one vertices and thus \cref{rule:deg-zero-one-vertices} might become applicable again.
To show our problem kernel, in the following theorem it is sufficient to first apply \cref{rule:deg-zero-one-vertices} exhaustively and afterwards apply \cref{rule:deg-two-vertices} exhaustively.

\begin{theorem}\label{thm:fes-lin-kernel}
	\Match admits a linear-time computable linear-size kernel with respect to the parameter \paramEnv{feedback edge number}~$k$.
\end{theorem}

\begin{proof}
	Let $G$ be the input graph. First we apply \cref{rule:deg-zero-one-vertices} to $G$ exhaustively, obtaining graph $G_1 = (V_1, E_1)$, and then we apply \cref{rule:deg-two-vertices} to~$G_1$ exhaustively, obtaining graph ~$G_2 = (V_2, E_2)$. 
	Note that both~$G_1$ and~$G_2$ can be computed in linear time~\cite{BK09a}.
	We will prove that~$G_2$ has at most~$6 \distPara$ vertices and~$7 \distPara$ edges.
	Denote with~$X_1 \subseteq E_1$ and~$X_2 \subseteq E_2$ the minimum feedback edge sets for~$G_1$ and~$G_2$ respectively. 
	Note that $|X_2| \le |X_1| \le \distPara$.
For any graph $H$, denote with~$V^1_{H}$, $V^2_{H}$, and~$V^{\ge 3}_{H}$ the vertices of $H$ 
that have degree one, two, and more than two, respectively (in our case $H$ will be replaced 
by~$G_1 - X_1$ or~$G_2 - X_2$, respectively). 
Observe that all vertices in~$G_1$ have degree at least two, since $G_1$ is reduced with respect to \cref{rule:deg-zero-one-vertices}.
	Thus $|V^1_{G_1-X_1}| \le 2\distPara$, as each leaf in~$G_1-X_1$ has to be incident to an edge in~$X_1$.
	Next, since~$G_1-X_1$ is a forest, we have that~$|V^{\ge 3}_{G_1-X_1}| < |V^1_{G_1-X_1}|$, 
and thus~$|V^{\ge 3}_{G_1-X_1}| < 2\distPara$.
	Note that the number of degree-two vertices in~$G_1$ cannot be upper-bounded by a function of~$\distPara$.
	However, observe that the exhaustive application of \cref{rule:deg-two-vertices} to~$G_1$ removes all vertices that have degree-two in~$G_1$ and possibly merges some of the remaining vertices.
	Thus, $G_2$ contains no vertices with degree two and thus, $|V^2_{G_2-X_2}| \le 2 \distPara$.
	Altogether, the number of vertices in~$G_2$ is at most $|V^1_{G_1-X_1}| + |V^2_{G_2-X_2}| + |V^{\ge 3}_{G_1-X_1}| \le 6 \distPara$.
	Since~$G_2-X_2$ is a forest, it follows that~$G_2$ has at most~$|V_2|+\distPara \le 7\distPara$ edges. \qed 
\end{proof}

Applying the~$O(m\sqrt{n})$-time algorithm for \Match~\cite{MV80} on the above kernel yields the following.

\begin{corollary}\label[corollary]{cor:fes-alg}
 	\Match can be solved in~$O(n + m + k^{1.5})$ time, where~$k$ is the feedback edge number.
\end{corollary}

\subsection{Parameter feedback vertex number} \label{sec:fvs-kernel} 
We next provide for \Match a kernel of size~$2^{O(k)}$ computable in~$O(kn)$ time where~$k$ is the \paramEnv{feedback vertex number}.
Using a known linear-time factor 4-approximation algorithm~\cite{BGNR98}, we can compute an approximate feedback vertex set and use it in our kernelization algorithm.

Roughly speaking, our kernelization algorithm extends the linear-time computable kernel with respect to the parameter \paramEnv{feedback edge set}.
Thus, \cref{rule:deg-zero-one-vertices,rule:deg-two-vertices} play an important role in the kernelization.
Compared to the other kernels presented in this paper, the kernel presented here comes at the price of higher running time~$O(kn)$ and bigger kernel size (exponential).
It remains open whether \Match parameterized by the \paramEnv{feedback vertex number} admits a linear-time computable kernel (possibly of exponential size), or whether it admits a polynomial kernel computable in~$O(kn)$ time.

Subsequently, we describe our kernelization algorithm which keeps in the kernel all vertices in the given feedback vertex set~$X$ and shrinks the size of~$G-X$.
Before doing so, we need some further notation.
In this section, we assume that each tree is rooted at some arbitrary (but fixed) vertex such that we can refer to the parent and children of a vertex.
A leaf in~$G-X$ is called a \emph{bottommost leaf} either if it has no siblings or if all its siblings are also leaves. (Here, bottommost refers to the subtree with the root being the parent of the considered leaf.)
The outline of the algorithm is as follows (we assume throughout this section that~$k < \log n$ since otherwise the input instance is already a kernel of size~$O(2^k)$):
\begin{enumerate}
	\item \label[step]{step:reduce-rule-1-2} Reduce~$G$ wrt.~\cref{rule:deg-two-vertices,rule:deg-zero-one-vertices}. %
	\item \label[step]{step:max-matching-G-X} Compute a maximum matching~$M_{G-X}$ in~$G-X$ 
		(where $X$ is a feedback vertex set that is computed by the linear-time 4-approximation algorithm~\cite{BGNR98}).
	\item \label[step]{step:only-leaves-free} Modify~$M_{G-X}$ in linear time such that only the leaves of~$G-X$ are free.
	\item \label[step]{step:bound-free-leaves} Upper-bound the number of free leaves in~$G-X$ by~$k^2$ (\cref{sec:kernel-fvs-steps1-to-4}).
	\item \label[step]{step:bound-bottommost-leaves} Upper-bound the number of bottommost leaves in~$G-X$ by~$O(k^2 2^k)$
	(\cref{sec:kernel-fvs-step5}).
	\item \label[step]{step:bound-long-paths} Upper-bound the degree of each vertex in~$X$ by~$O(k^2 2^k)$. Then, use \cref{rule:deg-two-vertices,rule:deg-zero-one-vertices} to provide the kernel of size~$2^{O(k)}$ 
	(\cref{sec:kernel-fvs-step6}).
\end{enumerate}
Whenever we reduce the graph at some step, we also show that the reduction is correct. 
That is, the given instance is a yes-instance if and only if the reduced one is a yes-instance.
The correctness of our kernelization algorithm then follows by the correctness of each step. 
We discuss in the following the details of each step.

\subsubsection{\cref{step:only-leaves-free,step:max-matching-G-X,step:reduce-rule-1-2,step:bound-free-leaves}} 
\label{sec:kernel-fvs-steps1-to-4}

In this subsection, we first discuss the straightforward \cref{step:only-leaves-free,step:max-matching-G-X,step:reduce-rule-1-2} and then turn to \cref{step:bound-free-leaves}.

\paragraph{\Cref{step:only-leaves-free,step:max-matching-G-X,step:reduce-rule-1-2}.}
As in \cref{sec:fes-kernel}, we perform \cref{step:reduce-rule-1-2} in linear time by first applying \cref{rule:deg-zero-one-vertices} and then \cref{rule:deg-two-vertices} using the algorithm due to \citet{BK09a}.
By \cref{lem:rule-deg-0-1-2-correct} this step is correct. 

A maximum matching~$M_{G-X}$ in \cref{step:max-matching-G-X} can be computed by repeatedly matching a free leaf to its neighbor and by removing both vertices from the graph (thus effectively applying \cref{rule:deg-zero-one-vertices} to~$G-X$).
Clearly, this can be done in linear time.

\cref{step:only-leaves-free} can be done in~$O(n)$ time by traversing each tree in~$G-X$ in a BFS manner starting from the root:
If a visited inner vertex~$v$ is free, then observe that all children are matched since~$M_{G-X}$ is maximum.
Pick an arbitrary child~$u$ of~$v$ and match it with~$v$. 
The vertex~$w$ that was previously matched to~$u$ is now free and since it is a child of~$u$, it will be visited in the future.
Observe that \cref{step:max-matching-G-X,step:only-leaves-free} do not change the graph but only the auxiliary matching~$M_{G-X}$, and thus the first three steps are correct.
The next observation summarizes the short discussion above.

\begin{observation}\label{obs:first-three-steps}
	\Cref{step:only-leaves-free,step:max-matching-G-X,step:reduce-rule-1-2} are correct and can be applied in linear time.
\end{observation}

\paragraph{\Cref{step:bound-free-leaves}.}
Our goal is to upper-bound the number of edges between vertices of~$X$ and~$V \setminus X$, since we can then use a simple analysis as for the parameter \paramEnv{feedback edge set}. 
Furthermore, recall that by \cref{obs:MatchingSizeGandG-X} the size of any maximum matching in~$G$ is at most~$k$ plus the size of~$M_{G-X}$.
Now, the crucial observation here is that, if a vertex~$x \in X$ has at least~$k$ neighbors $\{v_1, \ldots, v_k\}$ in~$V \setminus X$ which are free wrt.~$M_{G-X}$, 
then there exists a maximum matching where~$x$ is matched to one of $\{v_1, \ldots, v_k\}$ since at most~$k-1$ can be ``blocked'' by other matching edges.
Indeed, consider otherwise a maximum matching $M$ in which $x$ is not matched with any of $\{v_1, \ldots, v_k\}$. 
Then, since $|X|=k$, note that at most $k-1$ vertices among $\{v_1, \ldots, v_k\}$ are matched in $M$ with a vertex in $X$; 
suppose without loss of generality that $v_k$ is not matched with any vertex in $X$ (and thus $v_k$ is not matched at all in $M$). If $x$ is unmatched in $M$, 
then the matching $M\cup \{\{x,v_k\}\}$ has greater cardinality than $M$, a contradiction. 
Otherwise, if $x$ is matched in $M$ with a vertex $z$, then $M\cup\{\{x,v_k\}\} \setminus \{\{x,z\}\}$ is 
another maximum matching of $G$, in which $x$ is matched with a vertex among $\{v_1, \ldots, v_k\}$.
Formalizing this idea, we obtain the following data reduction rule.
\begin{rrule} \label{rule:k-free-neighbors-not-in-X}
	Let~$G = (V,E)$ be a graph, let $X \subseteq V$ be a vertex subset of size~$k$, and let $M_{G-X}$ be a maximum matching for~$G-X$.
	If there is a vertex~$x \in X$ with at least~$k$ free neighbors~$V_x = \{v_1, \ldots, v_k\} \subseteq V\setminus X$, then delete all edges from~$x$ to vertices in~$V \setminus V_x$.
\end{rrule}

We first show the correctness and then the running time of \cref{rule:k-free-neighbors-not-in-X}.

\begin{lemma} \label[lemma]{lem:rule-k-free-neighbors-correct}
	\cref{rule:k-free-neighbors-not-in-X} is correct. %
\end{lemma}
\begin{proof}
	Denote by~$s$ the size of a maximum matching in the input graph~$G = (V,E)$ and by~$s'$ the size of a maximum matching in the new graph~$G' = (V',E')$, where some edges incident to~$x$ are deleted. 
	We need to show that~$s = s'$.
	Since any matching in~$G'$ is also a matching in~$G$, we easily obtain~$s \ge s'$.

	It remains to show~$s \le s'$.
	To this end, let~$M_G := M^{\max}_G(M_{G-X})$ be a maximum matching for~$G$ with the maximum overlap with~$M_{G-X}$ (see \cref{sec:prelim}).
	If~$x$ is free wrt.~$M_G$ or if~$x$ is matched to a vertex~$v$ that is also in~$G'$ a neighbor of~$x$, then~$M_G$ is also a matching in~$G'$ ($M_G \subseteq E'$) and thus we have~$s \le s'$.
	Hence, consider the remaining case where~$x$ is matched to some vertex~$v$ such that~$\{v,x\} \notin E'$, that is, the edge~$\{v,x\}$ was deleted by \cref{rule:k-free-neighbors-not-in-X}.
	Hence, $x$ has~$k$ neighbors~$v_1, \ldots, v_k$ in~$V \setminus X$ such that each of these neighbors is free wrt.~$M_{G-X}$ and none of the edges~$\{v_i,x\}, i \in [k]$, was deleted.
	Observe that by the choice of~$M_G$, the graph $G(M_{G-X}, M_G)$ (the graph over vertex set~$V$ and the edges that are either in~$M_{G-X}$ or in~$M_G$, see \cref{sec:prelim}) contains exactly~$s - |M_{G-X}|$ paths of length at least one. 
	Each of these paths is an augmenting path for~$M_{G-X}$.
	By \cref{obs:MatchingSizeGandG-X}, we have~$s - |M_{G-X}| \le k$.
	Observe that~$\{v,x\}$ is an edge in one of these augmenting paths; denote this path with~$P$. 
	Thus, there are at most~$k-1$ augmenting paths in~$G(M_{G-X},M_G)$ that do not contain~$x$.
	Also, each of these paths contains exactly two vertices that are free wrt.~$M_{G-X}$: the endpoints of the path.
	This means that no vertex in~$X$ is an inner vertex on such a path.
	Furthermore, since~$M_{G-X}$ is a maximum matching, it follows that for each path at most one of these two endpoints is in~$V \setminus X$.
	Hence, at most~$k-1$ vertices of $v_1, \ldots, v_k$ are contained in the~$k-1$ augmenting paths of~$G(M_{G-X},M_G)$ except~$P$.
	Consequently, one of these vertices, say~$v_i$, is free wrt.~$M_G$ and can be matched with~$x$.
	Thus, by reversing the augmentation along~$P$ and adding the edge~$\{v_i,x\}$ we obtain another matching~$M'_G$ of size~$s$.
	Observe that~$M'_G$ is a matching for~$G$ and for~$G'$ and thus we have~$s \le s'$.
	This completes the proof of correctness.
	\qed
\end{proof}

\begin{lemma} \label[lemma]{lem:rule-k-free-neighbors-lin-time}
	\cref{rule:k-free-neighbors-not-in-X} can be exhaustively applied in~$O(n+m)$ time.
\end{lemma}

\begin{proof}
	We exhaustively apply the data reduction rule as follows.
	First, initialize for each vertex~$x \in X$ a counter with zero.
	Second, iterate over all free vertices in~$G-X$ in an arbitrary order.
	For each free vertex~$v \in V \setminus X$ iterate over its neighbors in~$X$.
	For each neighbor~$x \in X$ do the following: if the counter is less than~$k$, then increase the counter by one and mark the edge~$\{v,x\}$ (initially all edges are unmarked).
	Third, iterate over all vertices in~$X$.
	If the counter of the currently considered vertex~$x$ is~$k$, then delete all unmarked edges incident to~$x$.
	This completes the algorithm.
	Clearly, it deletes edges incident to a vertex~$x\in X$ if and only if~$x$ has~$k$ free neighbors in~$V \setminus X$ and the edges to these $k$~neighbors are kept.
	The running time is~$O(n+m)$: 
	When iterating over all free vertices in~$V \setminus X$ we consider each edge at most once.
	Furthermore, when iterating over the vertices in~$X$, we again consider each edge at most once. \qed
\end{proof}

To finish \cref{step:bound-free-leaves}, we exhaustively apply \cref{rule:k-free-neighbors-not-in-X} in linear time.
Afterwards, there are at most~$k^2$ free (wrt.\ to~$M_{G-X}$) leaves in~$G-X$ that have at least one neighbor in~$X$ since each of the~$k$ vertices in~$X$ is adjacent to at most~$k$ free leaves.
Thus, applying \cref{rule:deg-zero-one-vertices} we can remove the remaining free leaves that have no neighbor in~$X$. 
However, since for each degree-one vertex also its neighbor is removed, we might create new free leaves in~$G-X$ and need to again apply \cref{rule:k-free-neighbors-not-in-X} and update the matching (see \cref{step:only-leaves-free}).
This process of alternating application of \cref{rule:deg-zero-one-vertices,rule:k-free-neighbors-not-in-X} stops after at most~$k$ rounds since the neighborhood of each vertex in~$X$ can be changed by \cref{rule:k-free-neighbors-not-in-X} at most once.
This shows the running time $O(k(n+m))$.
We next show how to improve this to~$O(n+m)$.
In doing so, we arrive at the central proposition of this subsection, stating that \cref{step:only-leaves-free,step:max-matching-G-X,step:reduce-rule-1-2,step:bound-free-leaves} can be performed in linear time.

\begin{proposition}
\label[proposition]{prop:bound-free-leaves-lin-time}
	Given a matching instance~$(G,s)$ and a feedback vertex set~$X$, 
	\cref{alg:reduce-free-leaves-lin-time} computes in linear time an instance~$(G',s')$ with feedback vertex set~$X$ and a maximum matching~$M_{G'-X}$ in~$G'-X$ such that
the following holds.
	\begin{itemize}
		\item There is a matching of size~$s$ in~$G$ if and only if there is a matching of size~$s'$ in~$G'$.
		\item Each vertex in~$G'-X$ that is free wrt.~$M_{G'-X}$ is a leaf in~$G'-X$.
		\item There are at most~$k^2$ free leaves in~$G'-X$.  
	\end{itemize}
\end{proposition}

\begin{algorithm}[t!]\small
	\caption{
		An algorithm performing \cref{step:only-leaves-free,step:max-matching-G-X,step:reduce-rule-1-2,step:bound-free-leaves} in linear time.
	} 
	\label{alg:reduce-free-leaves-lin-time}
	\KwIn{A matching instance~$(G = (V,E),s)$ and a feedback vertex set~$X \subseteq V$ for~$G$ with $|X|=k$.}
	\KwOut{An equivalent matching instance~$(G',s')$ such that~$X$ is also a feedback vertex set for~$G'$ and a maximum matching~$M_{G'-X}$ for~$G'-X$ such that only at most~$k^2$ leaves in~$G'-X$ are free with respect to~$M_{G'-X}$.}
	\vspace{1em}
	Reduce~$G$ wrt.\ \cref{rule:deg-zero-one-vertices,rule:deg-two-vertices} \; \label{line:red-0-1-2}
	Compute a maximum matching~$M_{G-X}$ as described in \cref{step:only-leaves-free} \tcp*{see \cref{obs:first-three-steps}} \label{line:compute-M-G-X} 
	\lForEach(\label{line:initialize-counter}\tcp*[f]{$c(x)$ will store the number of free neighbors for~$x$}){$x \in X$}
	{
		$c(x) \gets $ 0 
	}
	\lForEach{$e \in E$}
	{
		marked($e$) $\gets$ False
	}
	$L \gets $ stack containing all free leaves in~$G-X$ in any order\; \label{line:initialize-leaf-stack}
	\While(\label{line:run-on-L}){$L$ is not empty}
	{
		$u \gets $ pop($L$)\;
		\ForEach(\label{line:start-check-k-free-neighbors}\tcp*[f]{Check if \cref{rule:k-free-neighbors-not-in-X} is applicable for~$x$}){$x \in N_G(u) \cap X$}
		{
			{
				$c(x) \gets c(x) + 1$, 
				marked($\{u,x\}$) $\gets$ True \label{line:mark-edge} \tcp*{fix~$u$ as free neighbor of~$x$}
			}
			\If(\label{line:enough-free-neighbors}\tcp*[f]{$x$ has enough free neighbors: apply \cref{rule:k-free-neighbors-not-in-X}}){$c(x) = k$}
			{
				\lForEach{$y \in N_G(x) \cap X$}
				{
					delete $\{x,y\}$
				}
				\ForEach{$v \in N_G(x) \setminus X$}
				{
					\If(\label{line:check-mark}){marked($\{x,v\}$) = False}
					{
						delete $\{x,v\}$\; \label{line:delete-edge-x-v}
						\tcc{Next deal with the case that~$v$ is a free leaf in~$G-X$ and~$x$ was the last neighbor of~$v$ in~$X$}
						\lIf{$\deg_G(v) = \deg_{G-X}(v) = 1$ and~$v$ is free}
						{
							push~$v$ on~$L$ \label{line:add-to-L-unmatched-deg-1-vertex}
						}
					}
				}
			}
		}
		\If(\label{line:deg-one-vertex}\tcp*[f]{$u$ has no neighbors in~$x$}){$\deg_G(u) = \deg_{G-X}(u) = 1$}
		{
			$v \gets $ neighbor of~$u$ in~$G-X$;
			$w \gets $ matched neighbor of~$v$ \label{line:def-w-v}\;
			delete~$u$ and~$v$ from~$G$ \label{line:delete-u-v} \tcp*{apply \cref{rule:deg-zero-one-vertices}}
			$M_{G-X} \gets M_{G-X} \setminus \{\{v,w\}\}$, $s \gets s-1$ \label{line:adjust-s}\tcp*{update~$M_{G-X}$ and~$s$}

			\tcc{augment along an arbitrary alternating path from~$w$ to a leaf in the subtree rooted in~$w$:}
			
			\While(\label{line:alternating-path-augment}){$w$ is not a leaf in~$G-X$}
			{
				$w' \gets $ arbitrary child of~$w$; $w'' \gets $ matched neighbor of~$w'$\; \label{line:select-child}
				$M_{G-X} \gets (M_{G-X} \setminus \{\{w',w''\}\}) \cup \{\{w,w'\}\}$ \; \label{line:augment}
				$w \gets w''$\; \label{line:update-w}
			}
			push~$w$ to~$L$ \label{line:push-w-to-L} \tcp*{$w$ is a free leaf, so add $w$ to the list of vertices to check}
		}
	}
	\KwRet{$(G,s)$ and~$M_{G-X}$}.
\end{algorithm}

Before proving \cref{prop:bound-free-leaves-lin-time}, we explain \cref{alg:reduce-free-leaves-lin-time} which reduces the graph with respect to \cref{rule:deg-zero-one-vertices,rule:k-free-neighbors-not-in-X} and updates the matching~$M_{G-X}$ as described in \cref{step:only-leaves-free}. %
The algorithm performs in \cref{line:red-0-1-2,line:compute-M-G-X} \cref{step:only-leaves-free,step:max-matching-G-X,step:reduce-rule-1-2}.
This can be done in linear time (see \cref{obs:first-three-steps}).
Next, \cref{rule:k-free-neighbors-not-in-X} is applied in \crefrange{line:start-check-k-free-neighbors}{line:add-to-L-unmatched-deg-1-vertex} using the approach described in the proof of \cref{lem:rule-k-free-neighbors-lin-time}:
For each vertex in~$x$ a counter~$c(x)$ is maintained.
When iterating over the free leaves in~$G-X$, these counters will be updated. 
If a counter~$c(x)$ reaches~$k$, then the algorithm knows that~$x$ has~$k$ fixed free neighbors and according to \cref{rule:k-free-neighbors-not-in-X} the edges to all other vertices can be deleted (see \cref{line:enough-free-neighbors}).
Observe that once the counter~$c(x)$ reaches~$k$, the vertex~$x$ will never be considered again by the algorithm since its only remaining neighbors are free leaves in~$G-X$ that already have been popped from the stack~$L$.
The only difference from the description in the proof of \cref{lem:rule-k-free-neighbors-lin-time} is that the algorithm reacts if some leaf~$v$ in~$G-X$ lost its last neighbor in~$X$ (see \cref{line:add-to-L-unmatched-deg-1-vertex}). %
If~$v$ is free, then add~$v$ to the stack~$L$ of unmatched degree-one vertices and defer dealing with~$v$ to a second stage of the algorithm (in \crefrange{line:deg-one-vertex}{line:push-w-to-L}).
(If~$v$ is matched, then we deal with~$v$ in~\cref{step:bound-long-paths}.)

\begin{figure}
	\begin{center}
		\tikzstyle{knoten}=[circle,draw=black,minimum size=18pt,inner sep=2pt]
		\begin{tikzpicture}[scale=0.75]
			\foreach[count=\i] \x / \y / \txt in {0/0/v,-1/1.5/u,-2/3/w,-3/1.5/{},-4/0/{},-3/-1.5/{},-2/-3/w''}
			{
				\node[knoten] (v\i) at (\x,\y) {$\txt$};
			}
			\foreach \i in {2,...,6}
			{
				\path (v\i) edge[-]  ($(v\i) + (-0.7,-1.05)$);
				\path (v\i) edge[-]  ($(v\i) + ( 0.0,-1.05)$);
				\path (v\i) edge[-]  ($(v\i) + ( 0.7,-1.05)$);
			}
			\foreach \i [remember=\i as \xi (initially 1)] in {2,...,7}
			{
				\ifthenelse{\isodd{\i}}
				{
					\path (v\i) edge[-,ultra thick]  (v\xi);
				}
				{
					\path (v\i) edge[-]  (v\xi);
				}
			}
			\tikzstyle{edge} = [color=black,opacity=.15,line cap=round, line join=round, line width=25pt]
			\begin{pgfonlayer}{background}
				\draw[edge] (v1.center) \foreach \i in {2,...,7}{ -- (v\i.center)};
			\end{pgfonlayer}
			
			\draw [->,snake=snake,line after snake=1mm] (1,0) -- (3,0);

			\begin{scope}[xshift = 8.5 cm]
				\foreach[count=\i] \x / \y / \txt in {-2/3/w,-3/1.5/{},-4/0/{},-3/-1.5/{},-2/-3/w''}
				{
					\node[knoten] (v\i) at (\x,\y) {$\txt$};
				}
				\foreach \i in {1,...,4}
				{
					\path (v\i) edge[-]  ($(v\i) + (-0.7,-1.05)$);
					\path (v\i) edge[-]  ($(v\i) + ( 0.0,-1.05)$);
					\path (v\i) edge[-]  ($(v\i) + ( 0.7,-1.05)$);
				}
				\foreach \i [remember=\i as \xi (initially 1)] in {2,...,5}
				{
					\ifthenelse{\isodd{\i}}
					{
						\path (v\i) edge[-]  (v\xi);
					}
					{
						\path (v\i) edge[-,ultra thick]  (v\xi);
					}
				}
				\tikzstyle{edge} = [color=black,opacity=.15,line cap=round, line join=round, line width=25pt]
				\begin{pgfonlayer}{background}
					\draw[edge] (v1.center) \foreach \i in {2,...,5}{ -- (v\i.center)};
				\end{pgfonlayer}
			\end{scope}

		\end{tikzpicture}
		
	\end{center}
	\caption{
		Dealing with new degree-one vertices occurring during the application of \cref{rule:k-free-neighbors-not-in-X} within \cref{step:bound-free-leaves}.
		Only vertices visited in the tree~$G-X$ in \crefrange{line:deg-one-vertex}{line:push-w-to-L} of \cref{alg:reduce-free-leaves-lin-time} are shown. 
		Further possible neighbors are indicated by edges.
		Left side: Vertex~$v$ is a free leaf in~$G-X$ (vertices in~$X$ are not illustrated). 
		The gray highlighted alternating path indicates where \cref{alg:reduce-free-leaves-lin-time} augments the maximum matching~$M_{G-X}$ in~$G-X$.
		Bold edges indicate edges in~$M_{G-X}$.
		Vertex~$w''$ is the leaf where the augmentation stops ($w''$ is matched, otherwise~$M_{G-X}$ would not be a maximum matching).
		Right side: Situation after \cref{alg:reduce-free-leaves-lin-time} augmentation.
		Vertex~$w''$ will be added to the list~$L$ and further processed.
	}
	\label{fig:step-4-illustration}
\end{figure}
We next discuss this second stage from \crefrange{line:deg-one-vertex}{line:push-w-to-L} (see \cref{fig:step-4-illustration} for an illustration):
Let~$u$ be an entry in~$L$ such that~$u$ has degree one in \cref{line:deg-one-vertex}, that is, $u$ is a free leaf in~$G-X$ and has no neighbors in~$X$.
Then, following~\cref{rule:deg-zero-one-vertices}, delete~$u$ and its neighbor~$v$ and decrease the solution size~$s$ by one (see \cref{line:delete-u-v,line:adjust-s}).
Let~$w$ denote the previously matched neighbor of~$v$.
Since~$v$ was removed, $w$ is now free.
If~$w$ is a leaf in~$G-X$, then we can simply add it to~$L$ and deal with it later.
If~$w$ is not a leaf, then we need to update~$M_{G-X}$ since only leaves are allowed to be free.
To this end, augment along an arbitrary alternating path from~$w$ to a leaf in the subtree with root~$w$ (see \cref{line:alternating-path-augment,line:select-child,line:augment,line:update-w}).
This is done as follows: 
Pick an arbitrary child~$w'$ of~$w$.
Let~$w''$ be the matched neighbor of~$w'$.
Observe that~$w''$ has to exist as if~$w'$ would be free, then~$\{w,w'\}$ could be added to~$M_{G-X}$; a contradiction to the maximality of~$M_{G-X}$.
Since~$w$ is the parent of~$w'$, it follows that~$w''$ is a child of~$w'$.
Now, remove~$\{w',w''\}$ from~$M_{G-X}$, add~$\{w',w\}$ and repeat the procedure with~$w''$ taking the role of~$w$.
Observe that the endpoint of this found alternating path, after augmentation, always is a free leaf.
Thus, this free leaf needs to be pushed to~$L$.
This completes the algorithm description.

The correctness of \cref{alg:reduce-free-leaves-lin-time} (stated in the next lemma) follows in a straightforward way from the above discussion.
For the formal proofs we introduce some notation.
We denote by~$G_i$ (respectively~$M_i$) the intermediate graph (respectively matching) stored by \cref{alg:reduce-free-leaves-lin-time} before the~$i^{\text{th}}$ iteration of the while loop in \cref{line:run-on-L}, that is, $G_1$ is the input graph and~$M_1$ is the initial matching computed in \cref{line:compute-M-G-X}.
The following observation is easy to see but useful in our proofs.

\begin{observation}\label{obs:alg1-subgraph}
	For each~$i \in \{1,\ldots,q\}$ where~$q$ is the number of iterations of the while loop in \cref{line:run-on-L}, we have that~$M_i$ is a maximum matching for~$G_i - X$.
	If~$i \ge 2$, then~$G_i$ is a subgraph of~$G_{i-1}$.
\end{observation}

\begin{lemma}\label{lem:alg-1-correct}
	\cref{alg:reduce-free-leaves-lin-time} is correct, that is, given a matching instance~$(G,s)$ and a feedback vertex set~$X$, it computes an instance~$(G',s')$ with feedback vertex set~$X$ and a maximum matching~$M_{G'-X}$ in~$G'-X$ such that:
	\begin{enumerate}
		\item There is a matching of size~$s$ in~$G$ if and only if there is a matching of size~$s'$ in~$G'$. \label{part:algo-1-correct}
		\item Each vertex in~$G'-X$ that is free wrt.~$M_{G'-X}$ is a leaf in~$G'-X$. \label{part:only-leaves-free}
		\item There are at most~$k^2$ free vertices in~$G'-X$.  \label{part:only-few-leaves-free}
	\end{enumerate}
\end{lemma}

\begin{proof} 
	\cref{obs:alg1-subgraph} implies that the returned graph~$G'$ is a subgraph of the input graph~$G$.
	Thus, $X$ is a feedback vertex set for both these graphs.
	Moreover, by  \cref{obs:alg1-subgraph}, $M_{G'-X}$ is a maximum matching for~$G'-X$.

	As to \ref{part:algo-1-correct}, observe that \cref{alg:reduce-free-leaves-lin-time} obtains~$G'$ from~$G$ by deleting edges in \cref{line:delete-edge-x-v} according to \cref{rule:k-free-neighbors-not-in-X} and by deleting vertices in \cref{line:delete-u-v} according to \cref{rule:deg-zero-one-vertices}.
	Thus, \ref{part:algo-1-correct} follows from the correctness of these data reduction rules (see \cref{lem:rule-deg-0-1-2-correct,lem:rule-k-free-neighbors-correct}).
	
	As to \ref{part:only-leaves-free}, observe that~$G-X$ is changed if and only if the matching~$M_{G-X}$ is changed accordingly (see \crefrange{line:deg-one-vertex}{line:push-w-to-L}).
	That is, after each deletion of vertices, the algorithm ensures that only leaves are free.
	Moreover, during the algorithm~$M_{G-X}$ is always a maximum matching for~$G-X$.
	
	As to \ref{part:only-few-leaves-free}, observe that any free leaf in~$G-X$ that is not removed needs to have a neighbor in~$X$ (see \cref{line:deg-one-vertex}).
	As \cref{rule:k-free-neighbors-not-in-X} is applied in \crefrange{line:start-check-k-free-neighbors}{line:add-to-L-unmatched-deg-1-vertex}, there are at most~$k^2$ such free leaves. 
	\qed
\end{proof}

We next show that \cref{alg:reduce-free-leaves-lin-time} runs in linear time.
To this end, we need a further technical statement. 

\begin{lemma}\label{lem:augmenting-paths-survive-algo1}
	In~$G_i$, let~$P$ be an even-length alternating path wrt.\ $M_i$ from a free leaf~$r$ to a matched inner vertex~$t$ of~$G_i-X$.
	Let~$u$ be the matched neighbor of~$t$.
	Then for each~$j \in \{1,\ldots, i\}$ there exists in~$G_j$ an even-length alternating path~$P'$ from~$t$ to a free leaf~$r'$ such that the neighbor of~$t$ on~$P'$ is either (i)~$u$, (ii) $t$'s parent, or (iii) a vertex not contained in~$G_i$.
\end{lemma}

\begin{proof}
	We prove the statement of the lemma by induction on~$i$.
	The base case~$i=1$ is trivial since~$G_1 = G$ and thus~$P' = P$.
	
	Now assume the statement is true for~$G_{i-1}$, $i \ge 2$. 
	We show that it holds for~$G_i$ as well.
	By \cref{obs:alg1-subgraph}, $G_{i}$ is a subgraph of~$G_{i-1}$ (and of~$G$). 
	Thus, the path~$P$ is also contained in~$G_{i-1}$ (and in~$G$).
	If~$r$ is a leaf in~$G_{i-1}-X$ and if~$M_i$ contains the same edges of~$P$ as~$M_{i-1}$, then~$P$ is an even-length augmenting path in~$G_{i-1}$ and the statement of the lemma follows from applying the induction hypothesis and \cref{obs:alg1-subgraph}. %
	Thus, assume that (a)~$r$ is not a leaf in~$G_{i-1}-X$ or (b)~$M_i$ does not contain the same edges of~$P$ as~$M_{i-1}$ (or both).
	
	We start with case (a) assuming that~$r$ is not a leaf in~$G_{i-1}-X$.
	Then in the~$(i-1)^{\text{st}}$ iteration of the while loop in \cref{line:run-on-L}, \cref{alg:reduce-free-leaves-lin-time} deleted the child~$r'$ of~$r$ and the child~$r''$ of~$r'$ in \cref{line:delete-u-v}.
	Moreover, $M_{i-1}$ contained the edge~$\{r,r'\}$ and~$r''$ was a free leaf in~$G_{i-1}-X$.
	Thus, extending~$P$ by the two vertices~$r',r''$ yields in~$G_{i-1}$ an even-length alternating path~$P^*$ from $t$ to the free leaf~$r''$ such that the neighbor of~$t$ on~$P^*$ is~$u$.
	Hence, the statement of the lemma follows from the induction hypothesis and \cref{obs:alg1-subgraph}.
	
	We next consider case (b), assuming that~$M_i$ and~$M_{i-1}$ do not contain the same edges of~$P$.
	Thus, in the~$(i-1)^{\text{st}}$ iteration of the while loop in \cref{line:run-on-L}, \cref{alg:reduce-free-leaves-lin-time} augmented along some alternating path in \cref{line:alternating-path-augment,line:select-child,line:augment,line:update-w}.
	Denote with~$Q$ this alternating path and let~$w_q$ be starting point of~$Q$, that is, $w_q$ is the vertex~$w$ in \cref{line:def-w-v}.
	Let~$v_Q,u_Q$ be the two deleted vertices  in \cref{line:def-w-v}.
	Let~$r_Q$ be the other endpoint of~$Q$, that is, $r_Q$ is a leaf in~$G_{i-1}$ and thus a free leaf in~$G_i$.
	Since~$M_i$ and~$M_{i-1}$ differ on~$P$, this implies that the two paths~$Q$ and~$P$ overlap.
	Let~$z$ be the vertex on~$P$ and on~$Q$ which is closest to~$r$.
	If~$z = r = r_q$, then~$P$ is a subpath of~$Q$ and in~$G_{i-1}$ there is an alternating path~$P^*$ from~$t$ to the free leaf~$u_Q$.
	(Here, $P^*$ is the part of~$Q$ that is not contained in~$P$.)
	Since the alternating path built in~\cref{line:alternating-path-augment,line:select-child,line:augment,line:update-w} is only extended by selecting child vertices, this implies that~$w_q = t$ or~$w_q$ is an ancestor of~$t$.
	Thus, the neighbor of~$t$ in~$P^*$ is either~$t$'s parent or~$v_Q$, that is, a child of~$t$ not contained in~$G_i$.
	Hence, the statement of the lemma follows from the induction hypothesis and \cref{obs:alg1-subgraph}.
	
	It remains to consider the case that~$z \neq r$.
	Let~$z_Q$ ($z_P$) be the neighbor of~$z$ that is on~$Q$ but not on~$P$ (on~$P$ but not on~$Q$); similarly let~$z_{PQ}$ be the neighbor of~$z$ that is on both~$P$ and~$Q$.
	Since~$Q$ is an alternating path either~$\{z,z_{Q}\}$ or~$\{z,z_{PQ}\}$ is in~$M_{i-1}$.
	
	First consider the case that~$\{z,z_{Q}\}$ is in~$M_{i-1}$.
	Then, since both the subpath of~$Q$ from~$z$ to~$u_Q$ and the subpath of~$P$ from~$z$ to~$r$ are alternating, we obtain an augmenting path from~$u_Q$ over~$z$ to~$r$.
	This is a contradiction to the maximality of~$M_{i-1}$.
	
	Second, consider the case that~$\{z,z_{PQ}\}$ is in~$M_{i-1}$.
	Thus, (after augmenting~$Q$) the edge~$\{z,z_{PQ}\}$ is not in~$M_{i}$.
	Moreover, as~$\{z,z_{Q}\}$ is in~$M_{i}$, the edge~$\{z,z_{P}\}$ is also not in~$M_i$.
	This contradicts the fact that~$P$ is an alternating path.
	\qed
\end{proof}

\begin{lemma}\label{lem:alg-1-lin-time}
	\Cref{alg:reduce-free-leaves-lin-time} runs in~$O(n+m)$ time.
\end{lemma}

\begin{proof}
	By \cref{obs:first-three-steps}, \cref{step:only-leaves-free,step:max-matching-G-X,step:reduce-rule-1-2} in \cref{line:red-0-1-2,line:compute-M-G-X} can be executed in linear time.
	Moreover, it is easy to execute \crefrange{line:initialize-counter}{line:initialize-leaf-stack} in one sweep over the graph, that is, in linear time.
	It remains to show that \crefrange{line:run-on-L}{line:push-w-to-L} run in linear time.
	To this end, we prove that each edge in~$E$ is being processed at most two times in \crefrange{line:run-on-L}{line:push-w-to-L}. 
	
	Start with the edges with at least one endpoint in~$X$.
	These edges will be inspected at most twice by the algorithm:
	Once, when the edge is marked (see \cref{line:mark-edge}).
	The second time is when the edge is checked and possibly deleted \cref{line:check-mark,line:delete-edge-x-v}.
	This shows that the first part (\crefrange{line:start-check-k-free-neighbors}{line:add-to-L-unmatched-deg-1-vertex}) runs in linear time.
	
	It remains to consider the edges within~$G-X$.
	To this end, observe that the algorithm performs two actions on the edges: deleting the edges (\cref{line:delete-u-v}) and finding and augmenting along an alternating path (\cref{line:alternating-path-augment,line:select-child,line:augment,line:update-w}).
	Clearly, after deleting an edge it will no longer be considered, so it remains to show that each edge is part of at most one alternating path in \cref{line:augment}.
	Assume toward a contradiction that the algorithm augments along an edge twice or more.
	From all the edges that are augmented twice or more let~$e \in E$ be one that is closest to the root of the tree containing~$e$, that is, there is no edge closer to a root. %
 	Let~$P_1$ and~$P_2$ be the first two augmenting paths containing~$e$. 
 	Assume without loss of generality  that the algorithm augmented along~$P_1$ in iteration~$i_1$ and along~$P_2$ in iteration~$i_2$ of the while loop in \cref{line:run-on-L} with~$i_1 \le i_2$.
	Let~$w_1$ and~$w_2$ be the two start points (the respective vertex~$w$ in \cref{line:def-w-v}) of~$P_1$ and~$P_2$ respectively.
	Let~$u_1$ and~$v_1$ ($u_2$ and~$v_2$) be the vertices deleted in \cref{line:delete-u-v} which in turn made~$w_1$ ($w_2$) free.
	Observe that~$e$ does not contain any of these four vertices~$u_1,v_1,u_2,v_2$ since before augmenting~$P_1$ ($P_2$) the vertices~$u_1$ and~$v_1$ ($u_2$ and~$v_2$) are deleted in \cref{line:delete-u-v}.
 	Since~$e$ is contained in both paths, either~$w_1$ is an ancestor of~$w_2$ or vice versa (or~$w_1 = w_2$).

 	Assume first that $w_2$ is an ancestor of~$w_1$.
 	Thus, $e = \{w_1,w_1'\}$ where~$w_1' \neq v_1$ and~$w_1'$ is a child of~$w_1$ (see \cref{line:select-child}).
 	Consider~$G_{i_2}$ and~$M_{i_2}$ before the augmentation along~$P_2$.
	Clearly, in~$G_{i_2}$ there is an alternating path of length two from~$w_2$ to the free leaf~$u_2$.
	Thus, by \cref{lem:augmenting-paths-survive-algo1}, in~$G_{i_1}$ there is an alternating path~$Q_1$ from~$w_2$ to a free leaf~$r$ such that~$r$ and~$w_1$ are not in the same subtree of~$w_2$.
	Moreover, by choice of~$e$ the two matchings~$M_{i_1}$ and~$M_{i_2}$ contain the same edges on the path from~$w_1$ to~$w_2$ in~$G-X$.
	Hence, there is an alternating path~$Q_2$ from~$w_1$ to~$w_2$ in~$G_{i_1}$.
	There is also an alternating path~$Q_3$ from~$w_1$ to the free leaf~$u_1$ in~$G_{i_1}$ (see \cref{line:def-w-v}).
	Combining~$Q_1, Q_2, Q_3$ gives an augmenting path from~$u_1$ to~$r$ in~$G_{i_1}$; a contradiction to the maximality of~$M_{i_1}$ (see \cref{obs:alg1-subgraph}).
	
	Next, consider the case that~$w_1 = w_2$.
	By choice of~$e$ we have that~$e = \{w_1, w'\}$ with~$w'$ being a child of~$w_1$ in~$G_{i_2}$ and~$w' \neq v_2$.
	Thus, after the augmentation along~$P_1$ the edge~$e$ is matched (see \cref{line:select-child}).
	This is a contradiction to the choice of~$P_2$ and the fact that~$\{w_2,v_2\} \in M_{i_2}$ (see \cref{line:def-w-v}).
	
	Finally, consider the case that~$w_1$ is an ancestor of~$w_2$.
	By choice of~$e$ we have that~$e = \{w_2, w_2'\}$ with~$w_2'$ being a child of~$w_2$ in~$G_{i_2}$ and~$w_2' \neq v_2$.
	From the argumentation used in the case~$w_1 = w_2$ above, we can infer that after augmenting~$P_1$ the edge~$e$ is not matched, thus~$e \notin M_{i_1+1}$ and~$e \in M_{i_1}$.
	Observe that in~$G_{i_2}$ there is a length-two alternating path from~$w_2$ to the free leaf~$u_2$.
	Thus, by \cref{lem:augmenting-paths-survive-algo1}, there is an even-length alternating path~$P$ from~$w_2$ to a free leaf in~$G_{i}$.
	Moreover, the (matched) neighbor~$w_2'$ of~$w_2$ in~$P$ is either (i)~$v_2$, (ii) the parent of~$w_2$, or (iii) a vertex not in~$G_{i_2}$.
	Since~$e \in M_{i_1}$, it follows that~$w_2'$ is the matched neighbor of~$w_2$ on~$P$.
	However, $w_2'$ is in~$G_{i_2}$, is neither the parent of~$w_2$ nor of~$v_2$, a contradiction.
	\qed
\end{proof}
\cref{prop:bound-free-leaves-lin-time} now follows from \cref{lem:alg-1-correct,lem:alg-1-lin-time}.

\subsubsection{\cref{step:bound-bottommost-leaves}}
\label{sec:kernel-fvs-step5}

In this step we reduce the graph in~$O(kn)$ time so that at most~$k^2 (2^k + 1)$ bottommost leaves will remain in the forest~$G-X$.
We will restrict ourselves to consider leaves that are matched with their parent vertex in~$M_{G-X}$ and that do not have a sibling.
We call these bottommost leaves \emph{interesting}. 
Any sibling of a bottommost leaf is by definition also a leaf.
Thus, at most one of these leaves (the bottommost leaf or one of its siblings) is matched with respect to~$M_{G-X}$ and all other leaves are free.
Recall that in the previous step we upper-bounded the number of free leaves with respect to~$M_{G-X}$ by~$k^2$.
Hence there are at most~$2 k^2$ bottommost leaves that are not interesting (each free leaf can be a bottommost leaf with a sibling matched to the parent).

Our general strategy for this step is to extend the idea behind \cref{rule:k-free-neighbors-not-in-X}:
We want to keep for each pair of vertices~$x,y \in X$ at most~$k$ different internally vertex-disjoint augmenting paths from~$x$ to~$y$. 
In this step, we only consider augmenting paths of the form~$x,u,v,y$ where~$v$ is a bottommost leaf and~$u$ is~$v$'s parent in~$G-X$.
Assume that the parent~$u$ of~$v$ is adjacent to some vertex~$x \in X$. 
Observe that in this case any augmenting path starting with the two vertices~$x$ and~$u$ has to continue to~$v$ and end in a neighbor of~$v$.
Thus, the edge~$\{x,u\}$ can be only used in augmenting paths of length three.
Furthermore, for different parent vertices~$u \neq u'$ the length-three augmenting paths are clearly internally vertex-disjoint.
If we do not need the edge~$\{x,u\}$ because we kept $k$ augmenting paths from~$x$ to each neighbor~$y \in N(v) \cap X$ already, then we can delete~$\{x,u\}$.
Furthermore, if we deleted the last edge from~$u$ to~$X$ (or $u$ had no neighbors in~$X$ in the beginning), then~$u$ is a degree-two vertex in~$G$ and can be removed by applying \cref{rule:deg-two-vertices}.
As the child~$v$ of~$u$ is a leaf in~$G-X$, it follows that~$v$ has at most~$k+1$ neighbors in~$G$.
We show below (\cref{lem:rule-deg-2-time}) that the application of \cref{rule:deg-two-vertices} to remove~$u$ takes~$O(k)$ time. 
As we remove at most~$n$ vertices, at most~$O(kn)$ time is spent on \cref{rule:deg-two-vertices} in this step.

We now show that, after a simple preprocessing, one application of \cref{rule:deg-two-vertices} in the algorithm above can indeed be performed in~$O(k)$ time.

\begin{lemma} \label[lemma]{lem:rule-deg-2-time}
	Let~$u$ be a leaf in the tree~$G-X$, $v$ be its parent, and let~$w$ be the parent of~$v$.
	If~$v$ has degree two in~$G$, then applying \cref{rule:deg-two-vertices} to~$v$ (deleting~$v$, merging~$u$ and~$v$, and setting~$s := s - 1$) can be done in~$O(k)$ time plus~$O(kn)$ time for an initial preprocessing.
\end{lemma}

\begin{proof}
	The preprocessing is to simply create a partial adjacency matrix for~$G$ with the vertices in~$X$ in one dimension and~$V$ in the other dimension.
	This adjacency matrix has size~$O(kn)$ and can clearly be computed in~$O(kn)$ time.
	
	Now apply \cref{rule:deg-two-vertices} to~$v$.
	Deleting~$v$ takes constant time.
	To merge~$u$ and~$w$ iterate over all neighbors of~$u$.
	If a neighbor~$u'$ of~$u$ is already a neighbor of~$w$, then decrease the degree of~$u'$ by one, otherwise add~$u'$ to the neighborhood of~$w$.
	Then, relabel~$w$ to be the new merged vertex~$uw$.

	Since~$u$ is a leaf in~$G-X$ and its only neighbor in~$G-X$, namely~$v$, is deleted, it follows that all remaining neighbors of~$u$ are in~$X$.
	Thus, using the above adjacency matrix, one can check in constant time whether~$u'$ is a neighbor of~$w$.
	Hence, the above algorithm runs in~$O(\deg(u)) = O(k)$ time. \qed
\end{proof}

The above ideas are used in \cref{alg:bound-bottommost-leaves} which we use for this step (\cref{step:bound-bottommost-leaves}). 
\begin{algorithm}[t!]\small
	\caption{An algorithm performing \cref{step:bound-bottommost-leaves} in~$O(kn)$ time.} 
	\label{alg:bound-bottommost-leaves}
	\KwIn{A matching instance~$(G = (V,E),s)$, a feedback vertex set~$X \subseteq V$ of size~$k$ for~$G$ with~$k < \log n$, and a maximum matching~$M_{G-X}$ for~$G-X$ with at most~$k^2$ free vertices in~$G-X$ that are all leaves.}
	\KwOut{An equivalent matching instance~$(G',s')$ such that~$X$ is also a feedback vertex set for~$G'$ and $G'-X$ is a tree with at most~$k^2 (2^k + 1)$ bottommost leaves, and a maximum matching~$M_{G'-X}$ for~$G'-X$ with at most~$k^2$ free vertices in~$G'-X$ that are all leaves.}
	\vspace{1em}
	
	Fix an arbitrary bijection~$f \colon 2^X \rightarrow \{1,\ldots,2^{k}\}$ \label{line:bijection-2^X-N} \;
	\ForEach{$v \in V \setminus X$}
	{
		Set~$f_X(v) \gets f(N(v) \cap X)$ \label{line:determine-set-Y} \tcp*{The number $f_X(v) < n$ can be read in constant time.}
	}
	Initialize a table~$\tab$ of size~$k \cdot 2^k$ with~$\tab[x,f(Y)] \gets 0$ for all~$x \in X, \emptyset \subsetneq Y \subseteq X$ \label{line:exp-size-table-init}\;
	$P \gets{}$ list containing all parents of interesting bottommost leaves \label{line:compute-list-P} \;
	\While{$P$ is not empty}
	{
		$u \gets{}$pop$(P)$\;
		$v \gets{}$child vertex of~$u$ in~$G-X$\;
		\ForEach(\label{line:check-neighbor-of-P-vertex}){$x \in N(u) \cap X$}
		{
			\If(\label{line:access-table}){$\tab[x,f_X(v)] < k$}
			{
				$\tab[x,f_X(v)] \gets \tab[x,f_X(v)] + 1$ \label{line:update-table}
			}
			\Else
			{
				delete~$\{x,u\}$ \label{line:delete-u-x}\;
			}
		}
		\If{$u$ has now degree two in~$G$}
		{
			Apply~\cref{rule:deg-two-vertices} to~$u$ \label{line:apply-deg-two-rule} \tcp*{This decreases~$s$ by one.}
			$vw \gets{}$vertex resulting from merging~$v$ and the parent~$w$ of~$u$\;
			\If{$vw$ is now an interesting bottommost leaf}
			{
				add the parent of~$vw$ to~$P$ %
				\label{line:add-contracted-vertex-to-P}
			}
		}
	}
	\KwRet{$(G,s)$ and~$M_{G-X}$}.
\end{algorithm}
The algorithm is explained in the proof of the following proposition stating the correctness and the running time of \cref{alg:bound-bottommost-leaves}.

\begin{proposition}
\label[proposition]{prop:bound-bottommost-leaves}
	Let~$(G = (V,E),s)$ be a \Match instance, let $X \subseteq V$ be a feedback vertex set, and let~$M_{G-X}$ be a maximum matching for~$G-X$ with at most~$k^2$ free vertices in~$G-X$ that are all leaves.
	Then, \cref{alg:bound-bottommost-leaves} computes in~$O(kn)$ time an instance~$(G',s')$ with feedback vertex set~$X$ and a maximum matching~$M_{G'-X}$ in~$G'-X$ such that the following holds. 
	\begin{itemize}
		\item There is a matching of size~$s$ in~$G$ if and only if there is a matching of size~$s'$ in~$G'$.
		\item There are at most~$2k^2 (2^k + 1)$ bottommost leaves in~$G'-X$.
		\item There are at most~$k^2$ free vertices in~$G'-X$ and they are all leaves.
	\end{itemize}
\end{proposition}

\begin{proof}
	We start with describing the basic idea of the algorithm.
	To this end, let~$\{u,v\} \in E$ be an edge such that~$v$ is an interesting bottommost leaf, that is, $v$ has no siblings and is matched to its parent~$u$ by~$M_{G-X}$.
	Counting for each pair~$x \in N(u) \cap X$ and~$y \in N(v) \cap X$ one augmenting path in a simple worst-case analysis gives~$O(k^2)$ time per edge, which is too slow for our purposes.
	Instead, we count for each pair consisting of a vertex~$x \in N(u) \cap X$ and a set~$Y = N(v) \cap X$ one augmenting path.
	In this way, we know that for each~$y \in Y$ there is one augmenting path from~$x$ to~$y$ without iterating through all~$y \in Y$. 
	This comes at the price of considering up to~$k 2^k$ such pairs. 
	However, we will show that we can do the computations in~$O(k)$ time per considered edge in~$G-X$. %
	The main reason for this improved running time is a simple preprocessing that allows for a bottommost vertex~$v$ to determine~$N(v) \cap X$ in constant time.
	
	The preprocessing is as follows (see \crefrange{line:bijection-2^X-N}{line:determine-set-Y}):
	First, fix an arbitrary bijection~$f$ between the set of all subsets of~$X$ to the numbers~$\{1,2,\ldots,2^k\}$.
	This can be done for example by representing a set~$Y \subseteq X = \{x_1, \ldots, x_k\}$ by a length-$k$ binary string (a number) where the~$i^\text{th}$ position is 1 if and only if~$x_i \in Y$.
	Given a set~$Y \subseteq X$ such a number can be computed in~$O(k)$ time in a straightforward way.
	Thus, \crefrange{line:bijection-2^X-N}{line:determine-set-Y} can be performed in~$O(kn)$ time.
	Furthermore, since we assume that~$k < \log n$ (otherwise the input instance has already at most~$2^k$ vertices), we have that~$f(Y) < n$ for each~$Y \subseteq X$.
	Thus, reading and comparing these numbers can be done in constant time.
	Furthermore, in \cref{line:determine-set-Y} the algorithm precomputes for each vertex the number corresponding to its neighborhood in~$X$.
	
	After the preprocessing, the algorithm uses a table~$\tab$ where it counts an augmenting path from a vertex~$x \in X$ to a set~$Y \subseteq X$ whenever a bottommost leaf~$v$ has exactly~$Y$ as neighborhood in~$X$ and the parent of~$v$ is adjacent to~$x$ (see \crefrange{line:exp-size-table-init}{line:add-contracted-vertex-to-P}).
	To do this in~$O(kn)$ time, the algorithm proceeds as follows:
	First, it computes in \cref{line:compute-list-P} the set~$P$ which contains all parents of interesting bottommost leaves.
	Clearly, this can be done in linear time.
	Next, the algorithm processes the vertices in~$P$.
	Observe that further vertices might be added to~$P$ (see \cref{line:add-contracted-vertex-to-P}) during this processing.
	Let~$u$ be the currently processed vertex of~$P$, let~$v$ be its child vertex, and let~$Y$ be the neighborhood of~$v$ in~$X$.
	For each neighbor~$x \in N(u) \cap X$, the algorithm checks whether there are already~$k$ augmenting paths between~$x$ and~$Y$ with a table lookup in~$\tab$ (see \cref{line:access-table}).
	If not, then the table entry is incremented by one (see \cref{line:update-table}) since~$u$ and~$v$ provide another augmenting path.
	If yes, then the edge~$\{x,u\}$ is deleted in \cref{line:delete-u-x} (we show below that this does not change the maximum matching size).
	If~$u$ has degree two after processing all neighbors of~$u$ in~$X$, then, by applying \cref{rule:deg-two-vertices}, we can remove~$u$ and merge its two neighbors~$v$ and~$w$.
	It follows from \cref{lem:rule-deg-2-time} that this application of \cref{rule:deg-two-vertices} can be done in~$O(k)$ time.
	Hence, one iteration of the while loop requires~$O(k)$ time and thus \cref{alg:bound-bottommost-leaves} runs in~$O(kn)$ time.

	Recall that all vertices in~$G-X$ that are free wrt.\ $M_{G-X}$ are leaves.
	Thus, the changes to~$M_{G-X}$ by applying \cref{rule:deg-two-vertices} in \cref{line:apply-deg-two-rule} are as follows: 
	First, the edge~$\{u,v\}$ is removed and second the edge~$\{w,q\}$ is replaced by~$\{vw,q\}$ for some~$q \in V$.
	Hence, the matching~$M_{G-X}$ after running \cref{alg:bound-bottommost-leaves} has still at most~$k^2$ free vertices and all of them are leaves.
	
	It remains to prove that 
	\begin{itemize}
		\item[(a)] the deletion of the edge~$\{x,u\}$ in \cref{line:delete-u-x} results in an equivalent instance and 
		\item[(b)] that the resulting instance has at most~$2k^2(2^k+1)$ bottommost leaves.
	\end{itemize}
	First, we show (a).
	To this end, assume towards a contradiction that the new graph~$G' := G - \{x,u\}$ has a smaller maximum matching than~$G$ (clearly, $G'$ cannot have a larger maximum matching).
	Thus, any maximum matching~$M_G$ for~$G$ has to contain the edge~$\{x,u\}$.
	This implies that the child~$v$ of~$u$ in~$G-X$ is matched in~$M_G$ with one of its neighbors (except~$u$):
	If~$v$ is free wrt.~$M_G$, then deleting~$\{x,u\}$ from~$M_G$ and adding~$\{v,u\}$ yields another maximum matching not containing~$\{x,u\}$, a contradiction.
	Recall that~$N(v) = \{u\} \cup Y$ where~$Y \subseteq X$ since~$v$ is a leaf in~$G-X$.
	Thus, each maximum matching~$M_G$ for~$G$ contains for some~$y \in Y$ the edge~$\{v,y\}$.
	Observe that \cref{alg:bound-bottommost-leaves} deletes~$\{x,u\}$ only if there are at least~$k$ other interesting bottommost leaves~$v_1, \ldots, v_k$ in~$G-X$ such that their respective parent is adjacent to~$x$ and~$N(v_i) \cap X = Y$ (see \crefrange{line:check-neighbor-of-P-vertex}{line:delete-u-x}).
	Since~$|Y| \le k$, it follows by the pigeonhole principle that at least one of these vertices, say~$v_i$, is not matched to any vertex in~$Y$.
	Thus, since~$v_i$ is an interesting bottommost leaf, it is matched to its only remaining neighbor: its parent~$u_i$ in~$G-X$.
	This implies that there is another maximum matching
	$$M'_{G} := (M_G \setminus \{\{v,y\},\{x,u\},\{u_i,v_i\}\}) \cup \{\{v_i,y\},\{x,u_i\},\{u,v\}\},$$
	a contradiction to the assumption that all maximum matchings for~$G$ have to contain~$\{x,u\}$.
	
	We next show (b) that the resulting instance has at most~$2k^2 (2^k + 1)$ bottommost leaves.
	To this end, recall that there are at most~$2k^2$ bottommost leaves that are not interesting (see discussion at the beginning of this subsection).
	Hence, it remains to upper-bound the number of interesting bottommost leaves.
	Observe that each parent~$u$ of an interesting bottommost leaf has to be adjacent to a vertex in~$X$ since otherwise~$u$ would have been deleted in \cref{line:apply-deg-two-rule}.
	Furthermore, after running \cref{alg:bound-bottommost-leaves}, each vertex~$x\in X$ is adjacent to at most~$k 2^k$ parents of interesting bottommost leaves (see \crefrange{line:access-table}{line:delete-u-x}).
	Thus, the number of interesting bottommost leaves is at most~$k^2 2^k$.
	Hence, the number of bottommost leaves is upper-bounded by~$2k^2 (2^k + 1)$. \qed
\end{proof}

\subsubsection{\cref{step:bound-long-paths}}
\label{sec:kernel-fvs-step6}

In this subsection, we provide the final step of our kernelization algorithm.
Recall that in the previous steps we have upper-bounded the number of bottommost leaves in~$G-X$ by~$O(k^2 2^k)$.
We also computed a maximum matching~$M_{G-X}$ for~$G-X$ such that at most~$k^2$ vertices are free wrt.~$M_{G-X}$ and all free vertices are leaves in~$G-X$.
Using this, we next show how to reduce~$G$ to a graph of size~$O(k^3 2^k)$.
To this end we need some further notation.
A leaf in~$G-X$ that is not bottommost is called a \emph{pendant}.
We define~$T$ to be the \emph{pendant-free tree (forest) of~$G-X$}, that is, the tree (forest) obtained from~$G-X$ by removing all pendants.
The next observation shows that~$G-X$ is not much larger than~$T$.
This allows us to restrict ourselves on giving an upper bound on the size of~$T$ instead of~$G-X$.

\begin{observation}\
	Let~$G-X$ be as described above with vertex set~$V \setminus X$ and let~$T$ be the pendant-free tree (forest) of~$G-X$ with vertex set~$V_T$.
	Then,~$|V \setminus X| \le 2|V_T| + k^2$.
\end{observation}
\begin{proof}
	Observe that~$V \setminus X$ is the union of all pendants in~$G-X$ and~$V_T$.
	Thus, it suffices to show that~$G-X$ contains at most~$|V_T| + k^2$ pendants.
	To this end, recall that we have a maximum matching for~$G-X$ with at most~$k^2$ free leaves.
	Thus, there are at most~$k^2$ leaves in~$G-X$ that have a sibling which is also a leaf since from two leaves with the same parent at most one can be matched.
	Hence, all but at most~$k^2$ pendants in~$G-X$ have pairwise different parent vertices.
	Since all these parent vertices are in~$V_T$, it follows that the number of pendants in~$G-X$ is~$|V_T| + k^2$. \qed
\end{proof}
We use the following observation to provide an upper bound on the number of leaves of~$T$.

\begin{observation}
\label[observation]{obs:pendant-free-forest-leaves}
	Let~$F$ be a forest, let~$F'$ be the pendant-free forest of~$F$, and let~$B$ be the set of all bottommost leaves in~$F$.
	Then, the set of leaves in~$F'$ is exactly~$B$.
\end{observation}

\begin{proof}
	First observe that each bottommost leaf of~$F$ is a leaf of~$F'$ since no bottommost leaf is removed and~$F'$ is a subgraph of~$F$.
	Thus, it remains to show that each leaf~$v$ in~$F'$ is a bottommost leaf in~$F$.
	
	We distinguish two cases of whether or not~$v$ is a leaf in~$F$: 
	First, assume that~$v$ is not a leaf in~$F$. 
	Thus, all of its child vertices have been removed. 
	Since we only remove pendants to obtain~$F'$ from~$F$ and since each pendant is a leaf, it follows that~$v$ is in~$F$ the parent of one or more leaves~$u_1, \ldots, u_\ell$. 
	Thus, by definition, all these leaves~$u_1, \ldots, u_\ell$ are bottommost leaves, a contradiction to the fact that they were deleted when creating~$F'$.
	
	Second, assume that~$v$ is a leaf in~$F$.
	If~$v$ is a bottommost leaf, then we are done. 
	Thus, assume that~$v$ is not a bottommost leaf and hence a pendant.
	However, since we remove all pendants to obtain~$F'$ from~$F$, it follows that~$v$ is not contained in~$F'$, a contradiction. \qed
\end{proof}

From \cref{obs:pendant-free-forest-leaves} it follows that the set~$B$ of bottommost leaves in~$G-X$ is exactly the set of leaves in~$T$.
In the previous step we reduced the graph such that~$|B| \le 2k^2 (2^k + 1)$ (see \cref{prop:bound-bottommost-leaves}). 
Thus, $T$ has at most~$2k^2 (2^k + 1)$ vertices of degree one and, since~$T$ is a tree (a forest), $T$ also has at most~$2k^2 (2^k + 1)$ vertices of degree at least three.
Let~$V_T^{2}$ be the vertices of degree two in~$T$ and let~$V_T^{\neq 2}$ be the remaining vertices in~$T$.
From the above it follows that~$|V_T^{\neq 2}| \le 4k^2 (2^k + 1)$.
Hence, it remains to upper-bound the size of~$V_T^{2}$.
To this end, we will upper-bound the degree of each vertex in~$X$ by~$O(k^2 2^k)$ and then use \cref{rule:deg-zero-one-vertices,rule:deg-two-vertices}. 
We will check for each edge~$\{x,v\} \in E$ with~$x \in X$ and~$v \in V \setminus X$ whether we ``need'' it.
This check will use the idea from the previous subsection where each vertex in~$X$ needs to reach each subset~$Y \subseteq X$ at most~$k$ times via an augmenting path.
Similarly as in the previous subsection, we want to keep ``enough'' of these augmenting paths. 
However, this time the augmenting paths might be long and different augmenting paths might overlap.
To still use the basic approach, we use the following lemma stating that we can still somehow replace augmenting paths.

\begin{lemma}\label[lemma]{lem:exchange-augmenting-paths}
	Let~$M_{G-X}$ be a maximum matching in the forest~$G-X$.
	Let~$P_{uv}$ be an augmenting path for~$M_{G-X}$ in~$G$ from~$u$ to~$v$.
	Let~$P_{wx}$, $P_{wy}$, and~$P_{wz}$ be three internally vertex-disjoint augmenting paths from~$w$ to~$x$, $y$, and~$z$, respectively, such that~$P_{uv}$ intersects all of them.
	Then, there exist two vertex-disjoint augmenting paths with endpoints~$u$, $v$, $w$, and one of the three vertices~$x$, $y$, and~$z$.
\end{lemma}

\begin{proof}
	Label the vertices in~$P_{uv}$ alternating as \emph{odd} or \emph{even} with respect to~$P_{uv}$ so that no two consecutive vertices have the same label, $u$ is \emph{odd}, and~$v$ is \emph{even}.
	Analogously, label the vertices in~$P_{wx}$, $P_{wy}$, and $P_{wz}$ as odd and even with respect to~$P_{wx}$, $P_{wy}$, and $P_{wz}$, respectively, so that~$w$ is always odd.
	Since all these paths are augmenting, it follows that each edge from an even vertex to its succeeding odd vertex is in the matching~$M_{G-X}$ and each edge from an odd vertex to its succeeding even vertex is not in the matching.
	Observe that~$P_{uv}$ intersects each of the other paths at least at two consecutive vertices, since every second edge must be an edge in~$M_{G-X}$.
	Since~$G-X$ is a forest and all vertices in~$X$ are free with respect to~$M_{G-X}$, it follows that the intersection of two augmenting paths is connected and thus a path.
	Since~$P_{uv}$ intersects the three augmenting paths from~$w$, it follows that at least two of these paths, say~$P_{wx}$ and~$P_{wy}$, have a ``fitting parity'', that is, in the intersections of~$P_{uv}$ with~$P_{wx}$ and with~$P_{wy}$ the even vertices with respect to~$P_{uv}$ are either even or odd with respect to \emph{both}~$P_{wx}$ and~$P_{wy}$.
	
	Assume without loss of generality that in the intersections of the paths the vertices have the same label with respect to the three paths (if the labels differ, then revert the ordering of the vertices in~$P_{uv}$, that is, exchange the names of~$u$ and~$v$ and change all labels on~$P_{uv}$ to their opposite).
	Denote with~$v^1_s$ and~$v^1_t$ the first and the last vertex in the intersection of~$P_{uv}$ and~$P_{wx}$.
	Analogously, denote with~$v^2_s$ and~$v^2_t$ the first and the last vertex in the intersection of~$P_{uv}$ and~$P_{wy}$.
	Assume without loss of generality that~$P_{uv}$ intersects first with~$P_{wx}$ and then with~$P_{wy}$.
	Observe that~$v^1_s$ and~$v^2_s$ are even vertices and~$v^1_t$ and~$v^2_t$ are odd vertices since the intersections have to start and end with edges in~$M_{G-X}$ (see \cref{fig:augmenting-paths-intersection} for an illustration).
	\begin{figure}
		\begin{center}
			\tikzstyle{knoten}=[circle,draw=black,minimum size=18pt,inner sep=2pt]
			\begin{tikzpicture}[scale=0.75]
				\foreach \i / \txt in {1/u,2/{v^1_s},3/{},4/{},5/{v^1_t},6/{v^2_s},7/{},8/{},9/{v^2_t},10/{v}}
				{
					\ifthenelse{\i<6}{
						\node[knoten] (v\i) at (1.5*\i, 0) {$\txt$};
					}
					{
						\node[knoten] (v\i) at (1.5*\i + 1, 0) {$\txt$};
					}
				}
				\node[knoten] (v11) at (6.5, 2) {$w$};
				\node[knoten] (v12) at (7.5, -1.5) {$x$};
				\node[knoten] (v13) at (14.5, -1.5) {$y$};
				
				\foreach \i / \j in {1/2,3/4,5/6,7/8,9/10,11/2,11/6,12/5,13/9}
				{
					\path (v\i) edge[dashed,-]  (v\j);
				}
				\foreach \i / \j in {2/3,4/5,6/7,8/9}
				{
					\path (v\i) edge[ultra thick,-]  (v\j);
				}
				
				\tikzstyle{edge} = [color=black,opacity=.15,line cap=round, line join=round, line width=25pt]
				\begin{pgfonlayer}{background}
					\draw[edge, line width=30pt] (v1.center) -- (v10.center) -- cycle;

					\draw[edge] (v11.center) -- (v2.center) -- (v5.center) -- (v12.center);

					\draw[edge] (v11.center) -- (v6.center) -- (v9.center) -- (v13.center);
					
					\draw[edge, line width=10pt] (v1.center) -- (v5.center) -- (v12.center);
					\draw[edge, line width=10pt] (v11.center) -- (v6.center) -- (v10.center);
				\end{pgfonlayer}
			\end{tikzpicture}
		\end{center}
		\caption{
			The situation in the proof of \cref{lem:exchange-augmenting-paths}.
			The augmenting path from~$u$ to~$v$ intersects the two augmenting paths~$P_{wx}$ and~$P_{wy}$ from~$w$ to~$x$ and~$y$,  respectively. Bold edges indicate edges in the matching, dashed edges indicate odd-length alternating paths starting with the first and last edge not being in the matching.
			The gray paths in the background highlight the different augmenting paths: the initial paths from~$u$ to~$v$, $w$ to~$x$, and~$x$ to~$y$ as well as the new paths from~$u$ to~$x$ and~$w$ to~$v$ as postulated by \cref{lem:exchange-augmenting-paths}.
		}
		\label{fig:augmenting-paths-intersection}
	\end{figure}
	For an arbitrary path~$P$ and for two arbitrary vertices~$p_1,p_2$ of $P$, denote by~$p_1-P-p_2$ the subpath of $P$ from~$p_1$ to~$p_2$.
	Observe that~$u-P_{uv}-v^1_t-P_{wx}-x$ and~$w-P_{wy}-v^2_t-P_{uv}-v$ are vertex-disjoint augmenting paths. \qed
\end{proof}

\paragraph{Algorithm description.}
We now provide the algorithm for \cref{step:bound-long-paths} (see \cref{alg:final-fvs-kernel-step} for pseudocode).
The algorithm uses the same preprocessing (see \crefrange{line:bijection-2^X-N}{line:determine-set-Y}) as \cref{alg:bound-bottommost-leaves}. 
Thus, the algorithm can determine whether two vertices have the same neighborhood in~$X$ in constant time. 
As in \cref{alg:bound-bottommost-leaves}, \cref{alg:final-fvs-kernel-step} uses a table~$\tab$ which has an entry for each vertex~$x \in X$ and each set~$Y \subseteq X$.
The table is filled in such a way that the algorithm detected for each~$y \in Y$ at least~$\tab[x,Y]$ internally vertex-disjoint augmenting paths from~$x$ to~$y$.
\begin{algorithm}[t!]\small
	\caption{An algorithm for computing \cref{step:bound-long-paths} in~$O(kn)$ time.} %
	\label{alg:final-fvs-kernel-step}
	\KwIn{A matching instance~$(G = (V,E),s)$, a feedback vertex set~$X \subseteq V$ of size~$k$ for~$G$ with~$k < \log n$ and at most~$k^2 (2^k+1)$ bottommost leaves in~$G-X$, and a maximum matching~$M_{G-X}$ for~$G-X$ with at most~$k^2$ free vertices in~$G-X$ that are all leaves.}
	\KwOut{An equivalent matching instance~$(G',s')$ such that~$G'$ contains at most~$O(k^3 2^k)$ vertices and edges.}

	Fix an arbitrary bijection~$f \colon 2^X \rightarrow \{1,\ldots,2^{k}\}$ \label{line:bijection-2^X-N-II} \;
	\ForEach{$v \in V \setminus X$}
	{
		Set~$f_X(v) \gets f(N(v) \cap X)$ \label{line:determine-set-Y-II} \tcp*{The number $f_X(v) < n$ can be read in constant time.}
	}
	Initialize a table~$\tab$ of size~$k \cdot 2^k$ with~$\tab[x,f(Y)] \gets 0$ for~$x \in X, \emptyset \subsetneq Y \subseteq X$\;
	$T \gets{}$pendant-free tree (forest) of~$G-X$\;
	$V_T^{\ge 3} \gets{}$vertices in~$T$ with degree~$\geq 3$ \label{line:get-deg-three-vertices-in-T}\;

	\ForEach(\label{line:bound-degree-of-x-loop}){$x \in X$}
	{
		\ForEach{$v \in N(x) \setminus X$}
		{
			\If(\label{line:test-keep-edge}\tcp*[f]{Is~$\{x,v\}$ needed for an augmenting path?}){Keep-Edge($x,v$)${}={}$false}
			{
				delete~$\{x,v\}$ \label{line:delete-x-v}
			}
		}
	}
	Exhaustively apply first \cref{rule:deg-zero-one-vertices} and then \cref{rule:deg-two-vertices} \label{line:apply-reduction-rules}\;
	\KwRet{$(G,s)$}.
	
	\smallskip  

	\Fn(\label{line:function-start}){Keep-Edge($x \in X, v \in V\setminus X$)}
	{
		\lIf{$v$ is free wrt.~$M_{G-X}$ or $v \in V_T^{\ge 3}$}
		{
			\KwRet{true} \label{line:v-is-free-or-split-vertex}
		}
		$w \gets{}$ matched neighbor of~$v$ in~$M_{G-X}$ \label{line:chose-w} \;
		\lIf{$w \in V_T^{\ge 3}$ or~$w$ is adjacent to a free leaf in~$G-X$}
		{
			\KwRet{true} \label{line:w-is-adjacent-to-free-or-split-vertex}
		}
		\If(\label{line:table-lookup}){$w$ has at least one neighbor in~$X$ and $\tab[x,f_X(w)] < 6k^2$}
		{
			$\tab[x,f_X(w)] \gets \tab[x,f_X(w)] + 1$\label{line:table-entry-increase} \; 
			\KwRet{true} \label{line:table-entry-increase-return}
		}
		\ForEach(\label{line:extend-alternating-path}){$u \in N(w) \setminus \{v\}$ that is matched wrt.~$M_{G-X}$ and fulfills~$\{u,x\} \notin E$}
		{
			\lIf(\label{line:extend-alternating-with-u}) {Keep-Edge($u,x$)${}={}$true}
			{
				\KwRet{true}
			}
		}
		\KwRet{false} \label{line:function-end}
	}
\end{algorithm}
The main part of the algorithm is the boolean function `Keep-Edge' in \crefrange{line:function-start}{line:function-end} which makes the decision on whether to delete an edge~$\{x,v\}$ for~$v\in V\setminus X$ and~$x \in X$.
The function works as follows for edge~$\{x,v\}$:
Starting at~$v$ the graph will be explored along possible augmenting paths until a ``reason'' for keeping the edge~$\{x,v\}$ is found or no further exploration is possible (see \cref{fig:step-6-keep-edge-explore} for an illustration).
\begin{figure}[t!]
	\begin{center}
		\tikzstyle{knoten}=[circle,draw=black,minimum size=18pt,inner sep=2pt]
		\begin{tikzpicture}[]
			\foreach \i in {1,...,10}
			{
				\node[knoten] (v\i) at (\i, 2) {$v_{\i}$};
			}
			\path (v1) edge[-,ultra thick]  ($(v1) + (-1, 1)$);
			\path (v1) edge[-]  ($(v1) + (-1, 0)$);
			\path (v1) edge[-]  ($(v1) + (-1,-1)$);

			\foreach \i [remember=\i as \xi (initially 1)] in {2,...,10}
			{
				\ifthenelse{\isodd{\i}}
				{
					\path (v\i) edge[-,ultra thick]  (v\xi);
				}
				{
					\path (v\i) edge[-]  (v\xi);
				}
			}
			
			\node[knoten] (x) at (5, 0) {$x$};
			\node[knoten] (y) at (8, 0) {$y$};
			\node[knoten] (z) at (9, 0) {$z$};
			
			\path (x) edge[-]  (v3);
			\path (x) edge[-]  (v4);
			\path (y) edge[-]  (v5);
			\path (z) edge[-]  (v5);
			\path (x) edge[-]  (v8);

			\tikzstyle{edge} = [color=black,opacity=.15,line cap=round, line join=round, line width=20pt]
			\begin{pgfonlayer}{background}
				\draw[edge] (x.center) \foreach \i in {3,2,1}{ -- (v\i.center)};
				\draw[edge] (x.center) \foreach \i in {4,5}{ -- (v\i.center)};
				\draw[edge] (x.center) \foreach \i in {8,9,10}{ -- (v\i.center)};
			\end{pgfonlayer}

		\end{tikzpicture}
	\end{center}
	\caption{ 
		Illustration of the graph exploration of the function Keep-Edge in \cref{alg:final-fvs-kernel-step}:
		The vertices $x$ and $y$ are vertices in the feedback vertex set~$X$. 
		The vertices~$v_1, \ldots, v_{10}$ are part of~$G-X$ where~$v_{10}$ is a free leaf.
		The matching~$M_{G-X}$ is denoted by the thick edges.
		Three alternating paths are highlighted; each path represents an exploration of Keep-Edge from~$x$ that returns true:
		First, the path via~$v_3$ ends in~$v_1$---a vertex with degree more than two in~$G-X$ (see \cref{line:v-is-free-or-split-vertex}).
		The second path via~$v_4$ ends in~$v_5$---a vertex connected to two vertices in~$X$ (here we assume that there are less than~$6k^2$ paths from~$x$ to vertices adjacent to~$y$ and~$z$; see \crefrange{line:table-lookup}{line:table-entry-increase-return}).
		The third path via~$v_8$ ends in the free leaf~$v_{10}$ (see \cref{line:v-is-free-or-split-vertex}).
	}
	\label{fig:step-6-keep-edge-explore}
\end{figure}

If the vertex~$v$ is free wrt.~$M_{G-X}$, then~$\{x,v\}$ is an augmenting path and we keep~$\{x,v\}$ (see \cref{line:v-is-free-or-split-vertex}).
Observe that in \cref{step:bound-free-leaves} (see \cref{prop:bound-free-leaves-lin-time}) we upper-bounded the number of free vertices by~$k^2$ and all these vertices are leaves.
Thus, we keep a bounded number of edges incident to~$x$ because the corresponding augmenting paths can end at a free leaf.
We provide the exact bound below when discussing the size of the graph returned by~\cref{alg:final-fvs-kernel-step}.
In \cref{line:v-is-free-or-split-vertex}, the algorithm also stops exploring the graph and keeps the edge~$\{x,v\}$ if~$v$ has degree at least three in~$T$.
The reason is to keep the graph exploration simple by following only degree-two vertices in~$T$. 
This ensures that the running time for exploring the graph from~$x$ does not exceed~$O(n)$.
Since the number of vertices in~$T$ with degree at least three is bounded (see discussion after \cref{obs:pendant-free-forest-leaves}), it follows that only a bounded number of such edges~$\{x,v\}$ are kept.

If~$v$ is not free wrt.~$M_{G-X}$, then it is matched with some vertex~$w$.
If~$w$ is adjacent to some leaf~$u$ in~$G-X$ that is free wrt.~$M_{G-X}$, then the path~$x,v,w,u$ is an augmenting path.
Thus, the algorithm keeps in this case the edge~$\{x,v\}$, see \cref{line:w-is-adjacent-to-free-or-split-vertex}.
Again, since the number of free leaves is bounded, only a bounded number of edges incident to~$x$ will be kept.
If~$w$ has degree at least three in~$T$, then the algorithm stops the graph exploration here and keeps the edge~$\{x,v\}$, see \cref{line:w-is-adjacent-to-free-or-split-vertex}.
Again, this is to keep the running time at~$O(kn)$ overall.

Let~$Y \subseteq X$ denote the neighborhood of~$w$ in~$X$. 
Thus the partial augmenting path~$x,v,w$ can be extended to each vertex in~$Y$.
Thus, if the algorithm did not yet find~$6k^2$ paths from~$x$ to vertices whose neighborhood in~$X$ is also~$Y$, then the table entry~$\tab[x,f_X(w)]$ (where $f_X(w)$ encodes the set~$Y = N(w) \cap X$) is increased by one and the edge~$\{x,v\}$ will be kept (see \cref{line:table-entry-increase,line:table-entry-increase-return}). 
(Here we need~$6k^2$ paths since these paths might be long and intersect with many other augmenting paths, see proof of \cref{prop:final-fvs-kernel-step} for the details of why~$6k^2$ is enough.) 
If the algorithm already found~$6k^2$ ``augmenting paths'' from~$x$ to~$Y$, then the neighborhood of~$w$ in~$X$ is irrelevant for~$x$ and the algorithm continues. 

In \cref{line:extend-alternating-path}, all above discussed cases to keep the edge~$\{x,v\}$ do not apply and the algorithm extends the partial augmenting part~$x,v,w$ by considering the neighbors of~$w$ except~$v$.
Since the algorithm dealt with possible extensions to vertices in~$X$ in \cref{line:table-entry-increase,line:table-entry-increase-return,line:table-lookup} and with extensions to free vertices in~$G-X$ in \cref{line:v-is-free-or-split-vertex}, it follows that the next vertex on this path has to be a vertex~$u$ that is matched wrt.~$M_{G-X}$.
Furthermore, since we want to extend a partial augmenting path from~$x$, we require that~$u$ is not adjacent to~$x$: 
otherwise the length-one path~$x,u$ would be another, shorter partial augmenting path from~$x$ to~$u$ and we do not need the currently stored partial augmenting path.

\paragraph{Statements on \cref{alg:final-fvs-kernel-step}.}

To show that \cref{alg:final-fvs-kernel-step} indeed performs \cref{step:bound-long-paths}, we need further lemmas.
For each edge~$\{x,z\}$ with~$x \in X$ and~$z \in V \setminus X$ we denote by~$P(x,z)$ the induced subgraph of~$G-X$ on the vertices that are explored in the function Keep-Edge when called in \cref{line:test-keep-edge} with~$x$ and~$z$.
More precisely, we initialize~$P(x,z) := \emptyset$.
Whenever the algorithm reaches \cref{line:v-is-free-or-split-vertex}, we add~$v$ to~$P(x,z)$.
Furthermore, whenever the algorithm reaches \cref{line:table-lookup}, we add~$w$ to~$P(x,z)$.
Similarly, when the recursive call in \cref{line:extend-alternating-with-u} returns true, then we add~$u$ to~$P(x,z)$ in the recursive call (with~$u$ taking the role of~$v$).

We next show that~$P(x,z)$ is a path with at most one additional pendant.

\begin{lemma}\label[lemma]{lem:explore-in-paths}
	Let~$x \in X$ and~$z \in V \setminus X$ be two vertices such that~$\{x,z\} \in E$.
	Then, $P(x,z)$ is either a path or a tree with exactly one vertex~$z'$ having more than two neighbors in~$P(x,z)$.
	Furthermore, $z'$ has degree exactly three and $z$~is a neighbor of~$z'$. 
\end{lemma}

\begin{proof}
	We first show that all vertices in~$P(x,z)$ except~$z$ and its neighbor~$z'$ have degree at most two in~$P(x,z)$.
	Observe that having more vertices than~$z$ and~$z'$ in~$P(x,z)$ requires \cref{alg:final-fvs-kernel-step} to reach \cref{line:extend-alternating-path}.
	
	Let~$w$ be the currently last vertex when \cref{alg:final-fvs-kernel-step} continues the graph exploration in \cref{line:extend-alternating-path}.
	Observe that the algorithm therefore dealt with the case that~$w$ has degree at least three in the pendant-free tree~$T$ in \cref{line:w-is-adjacent-to-free-or-split-vertex}.
	Thus, $w$ is either a pendant leaf in~$G-X$ or~$w \notin V_T^{\geq3}$ (that is, $w$ has degree at most two in~$T$).
	In the first case, there is no candidate to continue and the graph exploration stops.
	In the second case, $w$ has degree at most two in~$T$.
	
	We next show that any candidate~$u$ for continuing the graph exploration in \cref{line:extend-alternating-with-u} is not a leaf in~$G-X$.
	Assume toward a contradiction that~$u$ is a leaf in~$G-X$.
	Since the parent~$w$ of~$u$ is matched with some vertex~$v \neq u$ (this is how~$w$ is chosen, see \cref{line:chose-w}), it follows that~$u$ is not matched.
	This implies that the function `Keep-Edge' would have returned true in \cref{line:w-is-adjacent-to-free-or-split-vertex} and would not have reached \cref{line:extend-alternating-path}, a contradiction.
	Thus, the graph exploration follows only vertices in~$T$.
	Furthermore, the above argumentation implies that~$w$ is not adjacent to a leaf unless this leaf is its predecessor~$v$ in the graph exploration.
	
	We now have two cases: Either~$w$ is not adjacent to a leaf in~$G-X$ or~$v=z$ is a leaf and~$w=z'$ is its matched neighbor.
	In the first case, $w$ has at most one neighbor~$u \neq v$ since~$w \notin V_T^{\geq3}$.
	Hence, $w$ has degree two in~$P(x,z)$.
	In the second case, $w=z'$ has at most two neighbors~$u \neq v$ and~$u' \neq v$.
	Thus, $z'$ has degree at most three. \qed
\end{proof}

For~$x \in X$ let
$$ \mathcal{P}_x := \{P(x,v) \mid \{x,v\} \in E \wedge v \in V \setminus X\}$$
be the union of all induced subgraphs that \cref{alg:final-fvs-kernel-step} explores from~$x$.

\begin{lemma}\label[lemma]{lem:explore-disjoint-parts}
	There exists a partition of~$\mathcal{P}_x$ into~$\mathcal{P}_x = \mathcal{P}^A_x \cup \mathcal{P}^B_x$ such that all graphs within~$\mathcal{P}^A_x$ and within~$\mathcal{P}^B_x$ are pairwise disjoint. 
\end{lemma}

\begin{proof}
	Since~$G-X$ is a tree (or forest), $G-X$ is also bipartite. 
	Let~$A$ and~$B$ be its two color classes (so~$A \cup B = V \setminus X$).
	We define the two parts~$\mathcal{P}^A_x$ and~$\mathcal{P}^B_x$ as follows:
	A subgraph~$P \in \mathcal{P}_x$ is in~$\mathcal{P}^A_x$ if the neighbor~$v$ of~$x$ in~$P$ is contained in~$A$, otherwise~$P$ is in~$\mathcal{P}^B_x$.
	
	We show that all subgraphs in~$\mathcal{P}^A_x$ and~$\mathcal{P}^B_x$ are pairwise vertex-disjoint.
	To this end, assume toward a contradiction that two graphs~$P, Q \in \mathcal{P}^A_x$ share some vertex.
	(The case~$P, Q \in \mathcal{P}^B_x$ is completely analogous.)
	Let~$p_1$ and~$q_1$ be the first vertex in~$P$ and~$Q$ respectively, that is,~$p_1$ and~$q_1$ are adjacent to~$x$ in~$G$.
	Observe that~$p_1 \neq q_1$. %
	Let~$u \neq x$ be the first vertex that is in~$P$ and in~$Q$.
	By \cref{lem:explore-in-paths}, $P$ and~$Q$ are paths or trees with at most one vertex of degree more than two and this vertex has degree three and is the neighbor of~$p_1$ or~$q_1$, respectively.
	This implies together with~$q_1,p_1 \in A$ that either~$u = p_1$ or~$u = q_1$.
	Assume without loss of generality that~$u = p_1$.
	Since~$p_1 \in A$ and~$q_1 \in A$ and~$u$ is a vertex in~$Q$, it follows that \cref{alg:final-fvs-kernel-step} followed~$u$ in the graph exploration from~$q_1$ in \cref{line:extend-alternating-with-u}.
	However, this is a contradiction since the algorithm checks in \cref{line:extend-alternating-path} whether the new vertex~$u$ in the path is not adjacent to~$x$.
	Thus, all subgraphs in~$\mathcal{P}^A_x$ and~$\mathcal{P}^B_x$ are pairwise vertex-disjoint. \qed
\end{proof}

We next show that if~$\tab[x,f(Y)] = 6k^2$ for some~$x\in X$ and~$Y \subseteq X$ (recall that~$f$ maps~$Y$ to a number, see \cref{line:bijection-2^X-N-II}), then there exist at least~$3k^2$ internally vertex-disjoint augmenting paths from~$x$ to~$Y$.

\begin{lemma} \label[lemma]{lems:table-entry-corresponds-to-paths}
	If in \cref{line:table-lookup} of \cref{alg:final-fvs-kernel-step} it holds for~$x \in X$ and~$Y \subseteq X$ that~$\tab[x,f(Y)] = 6k^2$, then there exist in~$G$ wrt.~$M_{G-X}$ at least~$3k^2$ alternating paths from~$x$ to vertices~$v_1, \ldots, v_{3k^2}$ such that all these paths are pairwise vertex-disjoint (except~$x$) and~$N(v_i) \cap X = N(w) \cap X$ for all~$i \in [3k^2]$.
\end{lemma}

\begin{proof}
	Note that each time $\tab[x,f(Y)]$ is increased by one (see \cref{line:table-entry-increase}), the algorithm found a vertex~$w$ such that there is an alternating path~$P$ from~$x$ to~$w$ and~$N(w) \cap X = Y$.
	Furthermore, since the function Keep-Edge returns true in this case, the edge from~$x$ to its neighbor on~$P$ is not deleted in \cref{line:delete-x-v}.
	Thus, there exist at least~$6k^2$ alternating paths from~$x$ to vertices whose neighborhood in~$X$ is exactly~$Y$.
	By \cref{lem:explore-disjoint-parts}, it follows that at least half of these~$6k^2$ paths are vertex-disjoint. \qed
\end{proof}

The next lemma shows that \cref{alg:final-fvs-kernel-step} is correct and runs in~$O(kn)$ time. 

\begin{proposition}
\label[proposition]{prop:final-fvs-kernel-step}
	Let~$(G = (V,E),s)$ be a matching instance, let $X \subseteq V$ be a feedback vertex set of size~$k$ with~$k < \log n$ and at most~$2k^2 (2^k+1)$ bottommost leaves in~$G-X$, and let~$M_{G-X}$ be a maximum matching for~$G-X$ with at most~$k^2$ free vertices in~$G-X$ that are all leaves.
	Then, \cref{alg:final-fvs-kernel-step} computes in~$O(kn)$ time an equivalent instance~$(G',s')$ of size~$O(k^3 2^k)$.
\end{proposition}

\begin{proof}
	We split the proof into three claims, one for the correctness of the algorithm, one for the returned kernel size, and one for the running time.

	\begin{myclaim}
		The input instance~$(G,s)$ is a yes-instance if and only if the instance~$(G',s')$ produced by \cref{alg:final-fvs-kernel-step} is a yes-instance. 
	\end{myclaim}
	
	\begin{proofofclaim}
		Observe that the algorithm changes the input graph only in two lines: \cref{line:delete-x-v,line:apply-reduction-rules}.
		By \cref{lem:rule-deg-0-1-2-correct}, applying \cref{rule:deg-zero-one-vertices,rule:deg-two-vertices} yields an equivalent instance.
		Thus, it remains to show that deleting the edges in \cref{line:delete-x-v} is correct, that is, it does not change the size of a maximum matching.
		To this end, observe that deleting edges does not increase the size of a maximum matching.
		Thus, we need to show that the size of the maximum matching does not decrease.
		Assume toward a contradiction that it does.
		
		Let~$\{x,v\}$ be the edge whose deletion decreased the maximum matching size.
		Redefine~$G$ to be the graph before the deletion of~$\{x,v\}$ and~$G'$ to be the graph after the deletion of~$\{x,v\}$.
		Recall that \cref{alg:final-fvs-kernel-step} gets as additional input a maximum matching~$M_{G-X}$ for~$G-X$.
		Let~$M_G := M^{\max}_G(M_{G-X})$ be a maximum matching for~$G$ with the largest possible overlap with~$M_{G-X}$ and let~$G^M := G(M_{G-X},M_G) = (V, M_{G-X} \bigtriangleup M_G)$ (see \cref{sec:prelim}).
		Since~$\{x,v\} \in M_G \setminus M_{G-X}$ and~$x$ is free wrt.~$M_{G-X}$, it follows that there is a path~$P$ in~$G^M$ with one endpoint being~$x$.

		Recall (see \cref{sec:prelim}) that since~$P$ is a path in~$G^M$ it follows that~$P$ is an augmenting path for~$M_{G-X}$.
		Since all vertices in~$X$ are free wrt.~$M_{G-X}$, it follows that all vertices in~$P$ except the endpoints are in~$V \setminus X$.
		Let~$z$ be the second endpoint of this path~$P$.
		We call a vertex on~$P$ an even (odd) vertex if it has an even (odd) distance to~$x$ on~$P$. (So~$x$ is an even vertex and~$v$ and~$z$ are odd vertices).
		Observe that~$v$ is the only odd vertex in~$P$ adjacent to~$x$: 
		Otherwise there would be another augmenting path from~$x$ to~$z$ which only uses vertices from~$P$.
		This would imply the existence of another maximum matching that does not use~$\{x,v\}$, a contradiction.
		
		Let~$u$ be the neighbor of~$z$ in~$P$.
		Since no odd vertex on~$P$ except~$v$ is adjacent to~$x$, it follows that the graph exploration in the function Keep-Edge starting from~$x$ and~$v$ in \cref{line:test-keep-edge} either reached~$u$ or returned true before.
		If~$z \in V \setminus X$, then in both cases, the function Keep-Edge would have returned true in \cref{line:test-keep-edge} and \cref{alg:final-fvs-kernel-step} would not have deleted~$\{x,v\}$, a contradiction. 
		Thus, assume that~$z \in X$.
		Therefore, the function Keep-Edge considered the vertex~$u$ in \cref{line:table-lookup} but did not keep the edge~$\{x,v\}$.
		Thus, when considering~$u$, it holds that~$\tab[x,f_X(u)] = 6k^2$, where~$f_X(u)$ encodes~$Y := N(u) \cap X$ and~$z \in Y$.
		
		By \cref{lems:table-entry-corresponds-to-paths}, it follows that there are~$3k^2$ pairwise vertex-disjoint (except~$x$) alternating paths from~$x$ to vertices~$u_1, \ldots, u_{3k^2}$ with~$N(u_i) \cap X = Y$.
		Thus, there is a set~$\mathcal{Q}$ of~$3k^2$ internally vertex-disjoint paths from~$x$ to~$y$ in~$G$.
		If one of the paths~$Q \in \mathcal{Q}$ does not intersect any path in~$G^M$, then reverting the augmentation along~$P$ and augmenting along~$Q$ would result in another maximum matching not containing~$\{x,v\}$, a contradiction.
		Thus, assume that each path in~$\mathcal{Q}$ intersects at least one path in~$G^M$.
		
		For each two paths~$Q_1,Q_2 \in \mathcal{Q}$ that intersect the same path~$P'$ in~$G^M$ it holds that each further path~$P''$ in~$G^M$ can intersect at most one of~$Q_1$ and~$Q_2$:
		Assume toward a contradiction that~$P''$ does intersect both~$Q_1$ and~$Q_2$. 
		Since no path in~$G^M$ except~$P$ contains~$x$ and~$z$ it follows that all intersections between the paths are within~$G - X$.
		Since~$P'$ and~$P''$ are vertex-disjoint and~$Q_1$ and~$Q_2$ are internally vertex-disjoint, it follows that there is a cycle in~$G-X$, a contradiction to the fact that~$X$ is a feedback vertex set.
		
		Since~$3k^2 > 3k + k^2$, it follows from the pigeon hole principle that there is a path~$P' \in G^M$ that intersects at least three paths~$Q_1, Q_2, Q_3 \in \mathcal{Q}$ such that no further path in~$G^M$ intersects them.
		We can now apply \cref{lem:exchange-augmenting-paths} and obtain two vertex-disjoint augmenting paths~$Q$ and~$Q'$.
		Thus, reverting the augmentation along~$P$ and~$P''$ and augmenting along~$Q$ and~$Q''$ yields another maximum matching for~$G$ which does not contain~$\{x,v\}$, a contradiction.
	\end{proofofclaim}

	\begin{myclaim}
		The graph~$G'$ returned by \cref{alg:final-fvs-kernel-step} has~$O(k^3 2^k)$ vertices and edges. 
	\end{myclaim}
	
	\begin{proofofclaim}
		We first show that each vertex~$x \in X$ has degree~$O(k^2 2^k)$ in~$G'$.
		To this end, we need to count the number of neighbors~$v \in N(x) \setminus X$ where the function Keep-Edge returns true in \cref{line:test-keep-edge}.
		By \cref{lem:explore-in-paths}, the function Keep-Edge explores the graph along one or two paths (essentially growing from one starting point into two directions).
		Recall that~$\mathcal{P}_x$ denotes the subgraphs induced by the graph exploration of Keep-Edge for the neighbors of~$x$.
		By \cref{lem:explore-disjoint-parts} there is a partition of~$\mathcal{P}_x$ into~$\mathcal{P}^A_x$ and~$\mathcal{P}^B_x$ such that within each part the subgraphs are pairwise vertex-disjoint. 
		We consider the two parts independently.
		We start with bounding the number of graphs in~$\mathcal{P}^A_x$ where the function `Keep-Edge' returned true (the analysis is completely analogous for~$\mathcal{P}^B_x$).
		
		Since all explored subgraphs are disjoint and all free vertices in~$G-X$ wrt.~$M_{G-X}$ are leaves, it follows that \cref{alg:final-fvs-kernel-step} returned at most~$k^2$ times true in \cref{line:w-is-adjacent-to-free-or-split-vertex} due to~$w$ being adjacent to a free leaf in~$G-X$.
		Also, the algorithm returns at most~$k^2$ times true in \cref{line:v-is-free-or-split-vertex} due to $v$~being free.
		Furthermore, the algorithm returns at most~$6k^2 \cdot 2^k$ times true in \cref{line:table-entry-increase-return}.
		Finally, we show that the algorithm returns at most~$8k^2 \cdot(2^k - 1)$ times true in \cref{line:v-is-free-or-split-vertex,line:w-is-adjacent-to-free-or-split-vertex}, respectively.
		It follows from the discussion below \cref{obs:pendant-free-forest-leaves} that~$T$, the pendent-free tree of~$G-X$, has at most~$2k^2 (2^k + 1)$ leaves (denoted by~$V_T^{1}$) and~$2k^2 (2^k + 1)$ vertices of degree at least three (denoted by~$V_T^{\geq 3}$).
		Let~$V_T$ be the vertices of~$T$.
		Since~$T$ is a tree (or forest), it has more vertices than edges and hence 
		$$ \sum_{v \in V_T} \deg_T(v) < 2 |V_T|$$ 
		which implies
		$$ \sum_{v \in V_T^{\geq 3}} \deg_T(v) < 2 \cdot |V_T^{\geq 3}| + |V_T^{1}|.$$
		Thus, \cref{alg:final-fvs-kernel-step} returns at most~$2 \cdot |V_T^{\geq 3}| + |V_T^{1}| < 6 k^2 (2^k + 1)$ times true in \cref{line:w-is-adjacent-to-free-or-split-vertex} due to~$w$ being a vertex in~$V_T^{\geq 3}$.
		Also, \cref{alg:final-fvs-kernel-step} returns at most~$|V_T^{\geq 3}| \le 2 k^2 (2^k + 1)$ times true in \cref{line:v-is-free-or-split-vertex} due to~$v$ being a vertex in~$V_T^{\geq 3}$.
		
		Summarizing, considering the graph explorations in~$\mathcal{P}^A_x$, \cref{alg:final-fvs-kernel-step} returns at most
		$$ k^2 + k^2 + 6k^2 \cdot 2^k + 8 k^2 (2^k + 1) \in O(k^2 2^k) $$
		times true in the function Keep-Edge.
		Analogously, considering the graph explorations in~$\mathcal{P}^A_x$, \cref{alg:final-fvs-kernel-step} also returned at most~$O(k^2 2^k)$ times true.
		Hence, each vertex~$x \in X$ has degree at most~$O(k^2 2^k)$ in~$G'$.
		
		We now show that the exhaustive application of first \cref{rule:deg-zero-one-vertices} and then \cref{rule:deg-two-vertices} indeed results in a kernel of the claimed size.
		To this end, denote with~$V^1_{G'-X}$, $V^2_{G'-X}$, and~$V^{\ge 3}_{G'-X}$ the vertices that have degree one, two, and at least three in~$G'-X$.
		We have $|V^1_{G'-X}| \in O(k^3 2^k)$ since each vertex in~$X$ has degree at most~$O(k^2 2^k)$ and~$G'$ is reduced wrt.\ \cref{rule:deg-zero-one-vertices}. 
		Next, since~$G'-X$ is a forest (or tree), we have~$|V^{\ge 3}_{G'-X}| < |V^1_{G'-X}|$ and thus~$|V^{\ge 3}_{G'-X}| \in O(k^3 2^k)$.
		Finally, each degree-two vertex in~$G'$ needs at least one neighbor of degree at least three since~$G'$ is reduced with respect to \cref{rule:deg-two-vertices}.
		Thus, each vertex in~$V^2_{G'-X}$ is either incident to a vertex in~$X$ or adjacent to one of the at most~$O(k^2 2^k)$ vertices in~$G'-X$ that have degree at least three.
		Thus, $|V^2_{G'-X}| \in O(k^3 2^k)$.
		Summarizing, $G'$ contains at most~$O(k^3 2^k)$ vertices and edges.
	\end{proofofclaim}

	\begin{myclaim}
		\cref{alg:final-fvs-kernel-step} runs in~$O(kn)$ time. 
	\end{myclaim}
	
	\begin{proofofclaim}
		First, observe that \crefrange{line:bijection-2^X-N-II}{line:get-deg-three-vertices-in-T} can be done in~$O(kn)$ time:
		The preprocessing and table initialization can be done in~$O(kn)$ time as discussed in \cref{sec:kernel-fvs-step5}.
		Furthermore, $T$ and~$V_T^{\ge 3}$ can clearly be computed in~$O(n+m) \le O(kn)$ time.
		Second, applying \cref{rule:deg-zero-one-vertices} in~$O(n+m)$ time is straightforward and \citet{BK09a} showed how to apply \cref{rule:deg-two-vertices} in~$O(n+m)$ time.
		Thus, it remains to show that each iteration of the foreach-loop in \cref{line:bound-degree-of-x-loop} can be done in~$O(n)$ time.

		By \cref{lem:explore-disjoint-parts}, the graphs~$\mathcal{P}_x$ explored from~$x$ can be partitioned into two parts such that within each part all subgraphs are vertex-disjoint. 
		Thus, each vertex in~$G-X$ is visited only twice during the execution of the function Keep-Edge.
		Furthermore, observe that in \cref{line:table-lookup,line:table-entry-increase} the table can be accessed in constant time.
		Thus, the function Keep-Edge only checks once whether a vertex in~$V \setminus X$ has a neighbor in~$X$, namely in \cref{line:extend-alternating-path}.
		This single check can be done in constant time.
		Since the rest of the computation is done on~$G-X$ which has less than~$|V \setminus X|$ edges, it follows that each iteration of the foreach-loop in \cref{line:bound-degree-of-x-loop} can indeed be done in~$O(n)$ time.
	\end{proofofclaim}
	
\medskip

This completes the proof of~\cref{prop:final-fvs-kernel-step}. \qed
\end{proof}

This completes the description of \cref{step:bound-long-paths}.
Combining \crefrange{step:reduce-rule-1-2}{step:bound-long-paths} we obtain our kernelization algorithm for the parameter \paramEnv{feedback vertex number}.
\begin{theorem}
\label{thm:fvs-kernel}
	\Match parameterized by the \emph{feedback vertex number}~$k$ admits a kernel of size~$2^{O(k)}$.
	It can be computed in~$O(kn)$ time.
\end{theorem}

\begin{proof}
	First, using the linear-time factor-four approximation of \citet{BGNR98}, we compute an approximate feedback vertex set~$X$ with~$|X| \le 4k$.
	Then, we apply \crefrange{step:reduce-rule-1-2}{step:bound-long-paths} using
	\cref{alg:final-fvs-kernel-step,alg:bound-bottommost-leaves,alg:reduce-free-leaves-lin-time}.
	By \cref{prop:bound-free-leaves-lin-time,prop:bound-bottommost-leaves,prop:final-fvs-kernel-step}, this can be done in~$O(kn)$ time and results in a kernel of size~$O((4k)^3 2^{4k}) = 2^{O(k)}$. \qed
\end{proof}
Applying the~$O(m\sqrt{n})$-time algorithm for \Match~\cite{MV80} on the kernel yields the following.

\begin{corollary}\label[corollary]{cor:fvs-alg}
	\Match can be solved in~$O(kn + 2^{O(k)})$ time, where~$k$ is the feedback vertex number.
\end{corollary}

\section{Kernelization for Matching on Bipartite Graphs} \label{sec:bipartite-case} 

In this section, we investigate the possibility of efficient and effective preprocessing for \BipMatch.
More specifically, we show a linear-time computable polynomial-size kernel with respect to the parameter \paramEnv{distance to chain graphs}.
In the first part of this section, we provide the definition of chain graphs and describe how to compute the parameter.
In the second part, we discuss the kernelization algorithm.

\paragraph{Definition and computation of the parameter.}
We first define chain graphs which are a subclass of bipartite graphs with special monotonicity properties. 

\begin{definition}[\cite{BLS99}]\label[definition]{def:chain}
Let $G=(A,B,E)$ be a bipartite graph. Then $G$ is a \emph{chain graph} if each of its two color classes $A$, $B$  admits a linear ordering wrt.~neighborhood inclusion, that is, $A = \{a_1, \ldots, a_\alpha\}$ and~$B = \{b_1, \ldots, b_\beta\}$ where~$N(a_i) \subseteq N(a_j)$ and~$N(b_i) \subseteq N(b_j)$ whenever~$i < j$.
\end{definition}

Observe that if the graph~$G$ contains twins, then there is more than one linear ordering wrt.\ neighborhood inclusion.
To avoid ambiguities, we fix for the vertices of the color class $A$ (resp.~$B$) in a chain graph $G=(A,B,E)$ one linear ordering~$\prec_{A}$ (resp.~$\prec_{B}$) such that, for two vertices $u,v\in A$ (resp.~$u,v\in B$), if $u \prec_{A} v$ (resp.~if $u \prec_{B} v$) then $N(u) \subseteq N(v)$. 
In the remainder of the section we consider a bipartite representation of a given chain graph $G=(A,B,E)$ where the vertices of $A$ (resp.~$B$) are ordered according to~$\prec_{A}$ (resp.~$\prec_{B}$) from left to right (resp.~from right to left), as illustrated in \cref{fig:chain-graph}. 
\begin{figure}
	\centering
	\begin{tikzpicture}[scale=1]
		\tikzstyle{knoten}=[circle,draw,minimum size=17pt,inner sep=2pt,fill=white]

		\def\n{7}
		\def\xScale{1.25}
		
		\foreach \i in {1,...,\n} {
			\node[knoten] (a\i) at (\xScale * \i,2.5) {$a_{\i}$};
			\node[knoten] (b\i) at (\xScale * \n - \xScale * \i + \xScale,0) {$b_{\i}$};
		}
		\foreach \i in {2,...,\n} {
			\node at (\xScale * \i - 0.5 * \xScale,2.5) {$\prec_A$};
			\node at (\xScale * \i - 0.5 * \xScale,0) {$\succ_B$};
		}
		\node[anchor=east] at (\xScale * \n + \xScale,2.5) {$A$};
		\node[anchor=east] at (\xScale * \n + \xScale,0) {$B$};

		\foreach[count=\nr] \x in {7,7,5,5,5,2,1}
			\foreach \y in {\n, ..., \x} {
				\path (a\nr) edge[-,draw=black!80] (b\y);
			}
		
		\foreach \x in {3, ..., 7}
		{
			\pgfmathtruncatemacro{\y}{10 - \x};
			\path (a\x) edge[-,very thick] (b\y);
		}
		
	\end{tikzpicture}
	\caption{
		A chain graph. 
		Note that the ordering $\prec_A$ of the vertices in~$A$ is going from left to right while the ordering~$\succ_B$ of the vertices in~$B$ is going from right to left.
		The reason for these two orderings being drawn in different directions is that a maximum matching can be drawn as parallel edges, see e.\,g.\ the bold edges.
		In fact, \cref{alg:chain-graphs-linear} computes such matchings with the matched edges being parallel to each other.
	}
	\label{fig:chain-graph}
\end{figure}

For simplicity of notation we use in the following $\prec$ to denote the orderings~$\prec_{A}$ and~$\prec_{B}$ whenever the color class is clear from the context. 
Note that we use the direction left/right to indicate the ordering~$\prec$.
That is, for a vertex~$a' \in A$ to the right (left) of~$a \in A$ we have~$a \prec a'$ ($a' \succ a$).
In contrast, for a vertex~$b' \in B$ to the right (left) of~$b  \in B$ we have~$b \succ b'$ ($b \prec b'$).

We next show that we have a constant-factor approximation for the parameter and the corresponding vertex subset working in linear time.
To this end, we use the following characterization of chain graphs.
Here, $2K_2$ denote the one-regular graph on four vertices (with disjoint two edges).

\begin{lemma}[\cite{BLS99}] \label[lemma]{lem:bipartite-chain-characterization} 
	A bipartite graph is a chain graph if and only if it does not contain an induced~$2K_2$.
\end{lemma}

\begin{lemma}
\label[lemma]{lem:linTimeApproxDistChainGraphs}
	There is a linear-time factor-4 approximation for the problem of deleting a minimum number of vertices in a bipartite graph in order to obtain a chain graph.
\end{lemma}

\begin{proof}
	Let~$G=(A,B,E)$ be a bipartite graph. 
	We compute a set~$S \subseteq A \cup B$ such that~$G-S$ is a chain graph and~$S$ is at most four times larger than a minimum size of such a set.
	The algorithm iteratively tries to find a~$2K_2$ and deletes the four corresponding vertices until no further~$2K_2$ is found.
	Since in each~$2K_2$, by \cref{lem:bipartite-chain-characterization}, at least one vertex needs to be removed, the algorithm yields the claimed factor-4 approximation.
	
	The details of the algorithm are as follows:
	First, it initializes~$S = \emptyset$ and sorts the vertices in~$A$ and in~$B$ by their degree; the vertices in~$A = \{a_1, \ldots, a_\alpha\}$ in increasing order and the vertices in~$B = \{b_1, \ldots, b_\beta\}$ in decreasing order, that is,~$\deg(a_1) \le \ldots \le \deg(a_\alpha)$ and~$\deg(b_1) \ge \ldots \ge \deg(b_\beta)$.
	Since the degree of each vertex is at most~$\max\{\alpha, \beta\}$, this can be done in linear time with e.\,g.\ Bucket Sort. 
	At any stage the algorithm deletes all vertices of degree zero and all vertices which are adjacent to all vertices in the other partition.
	The deleted vertices are not added to~$S$ since these vertices cannot participate in a~$2K_2$.
	Next, the algorithm recursively processes the vertices in~$A$ in a nondecreasing order of their degrees.
	Let~$a \in A$ be a minimum-degree vertex and let~$b \in B$ be a neighbor of~$a$.
	Since~$b$ is not adjacent to all vertices in~$A$ (otherwise~$b$ would be deleted), there is a vertex~$a' \in A$ that is not adjacent to~$b$.
	Since~$\deg(a) \le \deg(a')$ it follows that~$a'$ has a neighbor~$b'$ that is not adjacent to~$a$.
	Hence, the four vertices~$a, a', b, b'$ induce only two edges: $\{a,b\}$ and~$\{a',b'\}$ and thus form a~$2K_2$.
	Thus, the algorithm adds the four vertices to~$S$, deletes them from the graph, and continues with a vertex in~$A$ that has minimum degree.

	As to the running time, we now show that, after the initial sorting, the algorithm considers each edge only twice:
	Selecting~$a$ and~$b$ as described above can be done in~$O(1)$ time.
	To select~$a'$, the algorithm simply iterates over all vertices in~$A$ until it finds a vertex that is not adjacent to~$b$. 
	In this way at most~$\deg(b)+1$ vertices are considered.
	Similarly, by iterating over the neighbors of~$a'$, one finds~$b'$.
	Hence, the edges incident to~$a$, $a'$, $b$, and $b'$ are used once to find the vertices and a second time when these vertices are deleted.
	Thus, using appropriate data structures, the algorithm runs in~$O(n+m)$ time. \qed
\end{proof}

\paragraph{Kernelization overview.}
In the rest of this section, we provide a linear-time computable kernel for \BipMatch with respect to the parameter vertex deletion distance $k$ to chain graphs. 
On a high level, our kernelization algorithm consists of two steps: 
First, we upper-bound by $O(k)$ the number of neighbors of each vertex in the deletion set. 
Second, we mark $O(k^2)$ special vertices and we use the monotonicity properties of chain graphs to upper-bound the number of vertices that lie between any two consecutive marked vertices, thus bounding the total size of the reduced graph to $O(k^3)$ vertices. 

\paragraph{Step 1.}
Let~$G = (A,B,E)$ be the bipartite input graph, where $V=A\cup B$, and let~$X \subseteq V$ be a vertex subset such that~$G-X$ is a chain graph. 
By \cref{lem:linTimeApproxDistChainGraphs}, we can compute an approximate~$X$ in linear time.
Our kernelization algorithm uses a specific maximum matching~$M_{G-X} \subseteq E$ in~$G-X$ with \cref{alg:chain-graphs-linear} where all edges in~$M_{G-X}$ are ``parallel'' and all matched vertices are consecutive in the ordering~$\prec_{A}$ and~$\prec_{B}$, see also \cref{fig:chain-graph}. 
\begin{algorithm}[t]\small
	\caption{ An algorithm computing a maximum matching~$M$ in the chain graph~$G$ such that all edges in~$M$ are ``parallel'' (see \cref{fig:chain-graph} for a visualization.) } 
	\label{alg:chain-graphs-linear}
	\KwIn{A chain graph~$G = (V,E)$, $V = A \cup B$, $A = \{a_1, \ldots, a_\alpha\}$ and $B = \{b_1, \ldots, b_\beta\}$ with~$N(a_i) \subseteq N(a_j)$ and $N(b_i) \subseteq N(b_j)$ for~$i < j$.}
	\KwOut{A maximum matching of~$G$ where all matched edges are parallel.}
	\vspace{1em}
	Compute the size~$\solSize$ of a maximum matching in~$G$ using an algorithm of~\citet{SY96}\;
	$M \gets \{ \{a_{\alpha - \solSize + 1}, b_\beta\}, \{a_{\alpha - \solSize + 2}, b_{\beta-1}\}, \ldots, \{a_{\alpha}, b_{\beta-\solSize+1}\}\}$ \;
	\KwRet{$M$}.
\end{algorithm}%
Since in convex graphs matching is linear-time solvable~\cite{SY96} and convex graphs are a superclass of chain graphs, this can be done in~$O(n+m)$ time.
We use~$M_{G-X}$ in our kernelization algorithm to obtain some local information about possible augmenting paths. 
For example, each augmenting path has at least one endpoint in~$X$.
Forming this into a data reduction rule, with $s$~denoting the size of a maximum matching, yields the following.

\begin{rrule}\label{rule:MatchingSizeBounds}
	If~$|M_{G-X}| \ge \solSize$, then return a trivial yes-instance; if~$\solSize > |M_{G-X}| + \distPara$, then return a trivial no-instance.
\end{rrule}
The correctness of \cref{rule:MatchingSizeBounds} follows from \cref{obs:MatchingSizeGandG-X}.

We will show next that there is a maximum matching~$M_G$ for~$G$ in which each vertex in~$X$ is either matched with another vertex in~$X$ or with a ``small-degree vertex'' in~$G-X$.
This means that an augmenting path starting at some vertex in~$X$ will ``enter'' the chain graph~$G-X$ in a small-degree vertex. 
We now formalize this concept.
Recall that $u \prec v$ implies $N(u) \subseteq N(v)$.
For a vertex~$x \in X$ we define~$N^{V \setminus X}_\text{small}(x)$ to be the set of the~$\distPara$ neighbors of~$x$ in~$V \setminus X$ with the smallest degree, formally,
$$ N^{V \setminus X}_\text{small}(x) := \{w \in N(x) \setminus X \mid \distPara > |\{u \in N(x) \setminus X \mid u \prec w\}|\}.$$ 

\begin{lemma}\label[lemma]{lem:XMatchedToSmallestDegree}
	Let~$G=(V,E)$ be a bipartite graph and let~$X \subseteq V$ be a vertex set such that~$G-X$ is a chain graph.
	Then, there exists a maximum matching~$M_G$ for~$G$ such that every matched vertex~$x \in X$ is matched to a vertex in~$N^{V \setminus X}_\text{\rm small}(x) \cup X$.
\end{lemma}

\begin{proof}
	Assume, towards a contradiction, that there is no such matching~$M_G$. 
	Let~$M'_G$ be a maximum matching for~$G$ that maximizes the number of vertices~$x \in X$ that are matched to a vertex in~$N^{V \setminus X}_\text{small}(x) \cup X$, that is, let $M'_G$ maximize~$|\{x \in X \mid \{u,x\} \in M'_G \wedge u \in N^{V \setminus X}_\text{small}(x) \cup X\}|$.
	Let~$x \in X$ be a vertex that is not matched with any vertex in~$N^{V \setminus X}_\text{small}(x) \cup X$, that is, $x$ is matched to a vertex~$u \in V \setminus (N^{V \setminus X}_\text{small}(x) \cup X)$.
	If there is an unmatched vertex~$w \in N^{V \setminus X}_\text{small}(x)$ in~$M'_G$, then the matching~$M''_G := M'_G \cup \{\{x,w\}\} \setminus \{\{u,x\}\}$ is a maximum matching with more vertices $x \in X$ (compared to~$M'_G$) that are matched to a vertex in~$N^{V \setminus X}_\text{small}(x) \cup X$, a contradiction.
	Hence, assume that there is no free vertex in~$N^{V \setminus X}_\text{small}(x)$.
	Since~$|N^{V \setminus X}_\text{small}(x)|=|X|=k$, it follows that at least one vertex~$w \in N^{V \setminus X}_\text{small}(x)$ is matched to a vertex~$v \in V \setminus X$.
	Observe that, by definition of~$N^{V \setminus X}_\text{small}(x)$, we have~$N_{G-X}(w) \subseteq N_{G-X}(u)$.
	Thus, we have $\{u,v\} \in E$ and thus, $M''_G := M'_G \cup \{\{x,w\},\{u,v\}\} \setminus \{\{u,x\},\{w,v\}\}$ is a maximum matching with more vertices in~$X$ (compared to~$M'_G$) fulfilling the condition of the lemma, a contradiction. \qed
\end{proof}

Based on \cref{lem:XMatchedToSmallestDegree}, we can provide our next data reduction rule.

\begin{rrule}\label{rule:reduce-neighborhood-of-X}
	Let~$(G,s)$ be an instance reduced with respect to \cref{rule:MatchingSizeBounds} and let~$x \in X$.
	Then delete all edges between~$x$ and~$V \setminus N^{V \setminus X}_\text{\rm small}(x)$.
\end{rrule}

Clearly, \cref{rule:reduce-neighborhood-of-X} can be exhaustively applied in~$O(n+m)$ time by one iteration over~$A$ and~$B$ in the ordering~$\prec$.

\paragraph{Step 2.}
For the second step of our kernelization algorithm, we first mark a set~$K$ of~$O(k^2)$ vertices that are kept in the graph (and thus will end up in the kernel):
Keep all vertices of~$X$. 
For each vertex~$x \in X$ keep all vertices in~$N^{V \setminus X}_\text{small}(x)$ and if a kept vertex is matched wrt.~$M_{G-X}$, then keep also the vertex with which it is matched. 
Formally, we have:
$$ K:= X \cup \{v \mid \exists x \in X\colon v \in N^{V \setminus X}_\text{small}(x) \vee (\{u,v\} \in M_{G-X} \wedge u \in N^{V \setminus X}_\text{small}(x))\}.$$
Observe that exhaustively applying \cref{rule:reduce-neighborhood-of-X} ensures that~$K$ is of size at most~$2k^2$.

Next, we use the monotonicity properties of the chain graph to show that it suffices to keep for each vertex~$v \in K$ at most~$k$ vertices to the right and to the left of~$v$.
Consider an augmenting path~$P = x, a_1, b_1, \ldots, a_\ell, b_\ell, y$ from a vertex~$x \in B \cap X$ to a vertex~$y \in A \cap X$.
Observe that if~$a_1 \prec a_\ell$, then also~$\{b_1,a_\ell\} \in E$ and thus $P' = x, a_1, b_1, a_\ell,b_\ell,y$ is an augmenting path (see \cref{fig:chain-graph-augmenting-path-easy} for a visualization).
\begin{figure}
	\centering
	\begin{tikzpicture}[scale=1]
		\tikzstyle{knoten}=[circle,draw,minimum size=17pt,inner sep=2pt,fill=white]

		\def\n{7}
		\def\xScale{1.25}
		
		\foreach \i in {1,...,\n} {
			\node[knoten] (a\i) at (\xScale * \i,2.5) {$a_{\i}$};
			\node[knoten] (b\i) at (\xScale * \i,0) {$b_{\i}$};
		}
		
		\node[knoten] (x) at (0,0) {$x$};
		\node[knoten] (y) at (\xScale * \n + \xScale,2.5) {$y$};

		\foreach[count=\nr] \x in {2,3,5,6,7,7,7}
			\foreach \y in {1, ..., \x} {
				\path (a\nr) edge[-,draw=black!70] (b\y);
			}
		\path (x) edge[-,draw=black!70] (a1);
		\path (y) edge[-,draw=black!70] (b\n);
		
		\foreach \x in {1, ..., 7}
		{
			\path (a\x) edge[-,ultra thick] (b\x);
		}
		\tikzstyle{edge} = [color=black,opacity=.15,line cap=round, line join=round, line width=25pt]
		\begin{pgfonlayer}{background}
			\draw[edge, line width=10pt] (x.center) -- (a1.center) -- (b1.center) -- (a\n.center) -- (b\n.center) -- (y.center);
		\end{pgfonlayer}
	\end{tikzpicture}
	\caption{
		A chain graph with a maximum matching (thick edges) and two additional vertices~$x$ and~$y$.
		An augmenting path~$P = x, a_1, b_1, a_2, b_2, \ldots, a_7, b_7, y$ with~$a_1 \prec a_\ell$ implies that there is a shorter augmenting path~$P' = x, a_1, b_1, a_7, b_7, y$ of length five (indicated by the gray background) in the input graph since~$N(a_1) \subseteq N(a_7)$.
	}
	\label{fig:chain-graph-augmenting-path-easy}
\end{figure}
Furthermore, the vertices in the augmenting path~$P'$ are a subset of~$K \cup X$ and, thus, by keeping these vertices (and the edges between them), we also keep the augmenting path~$P'$ in our kernel.
Hence, it remains to consider the more complicated case that~$a_\ell \prec a_1$ (see \cref{fig:chain-graph-augmenting-path-difficult}).
\begin{figure}
	\centering
	\begin{tikzpicture}[scale=1]
		\tikzstyle{knoten}=[circle,draw,minimum size=17pt,inner sep=2pt,fill=white]

		\def\n{7}
		\def\xScale{1.25}
		
		\foreach \i in {1,...,\n} {
			\pgfmathtruncatemacro{\nr}{\n + 1 - \i}
			\node[knoten] (a\i) at (\xScale * \i,2.5) {$a_{\nr}$};
			\node[knoten] (b\i) at (\xScale * \i,0) {$b_{\nr}$};
		}
		
		\node[knoten] (x) at (0,2.5) {$y$};
		\node[knoten] (y) at (\xScale * \n + \xScale,0) {$x$};

		\foreach[count=\nr] \x in {2,3,5,6,7,7,7}
			\foreach \y in {1, ..., \x} {
				\path (a\nr) edge[-,draw=black!70] (b\y);
			}
		\path (x) edge[-,draw=black!70] (b1);
		\path (y) edge[-,draw=black!70] (a7);
		
		\foreach \x in {1, ..., 7}
		{
			\path (a\x) edge[-,ultra thick] (b\x);
		}
		\tikzstyle{edge} = [color=black,opacity=.15,line cap=round, line join=round, line width=25pt]
		\begin{pgfonlayer}{background}
			\draw[edge, line width=10pt] (x.center) -- (b1.center) -- (a1.center) %
			\foreach \i [remember=\i as \xi (initially 1)] in {2,3,5,7} {%
				-- (a\xi.center) -- (b\i.center) -- (a\i.center)%
			}
			-- (a\n.center) -- (y.center);
			
			\draw[rounded corners,dashed] ($(b5.north west)+(-0.25,0.25)$) rectangle ($(b6.south east)+(0.25,-0.25)$);
			\node at ($(b5.south east)+(0.4,-0.5)$) {$\rmv(a_4,M) = 2$};
			\draw[rounded corners,dashed] ($(a3.north west)+(-0.25,0.25)$) rectangle ($(a3.south east)+(0.25,-0.25)$);
			\node at ($(a3.north)+(0,0.4)$) {$\lmv(b_4,M) = 1$};
		\end{pgfonlayer}
	\end{tikzpicture}
	\caption{
		The graph from \cref{fig:chain-graph-augmenting-path-easy} with the only difference being that the positions of~$x$ and~$y$ are exchanged (and the vertex names are adjusted accordingly, so~$x$ is still adjacent to~$a_1$).
		Again, the thick, black edges denote a maximum matching~$M$ for the chain graph (containing all vertices except~$x$ and~$y$).
		The graph contains an augmenting path~$P = x, a_1, b_1, a_2, b_2, \ldots, a_7, b_7, y$ with~$a_1 \succ a_7$.
		In contrast to the example displayed in \cref{fig:chain-graph-augmenting-path-easy}, there is no augmenting~$x$--$y$-path of length five.
		The shortest augmenting~$x$--$y$-path is displayed.
		For the edge~$\{a_4,b_4\} \in M$ the vertices certifying that~$\lmv(b_4,M) = 1$ and~$\rmv(a_4,M) = 2$ are highlighted by dashed boxes (see \cref{def:lmv-rmv}).
	}
	\label{fig:chain-graph-augmenting-path-difficult}
\end{figure}
To this end, we next show that in certain ``areas'' of the chain graph~$G-X$ the number of augmenting paths ``passing through'' such an area is upper-bounded. 
To specify an ``area'', we need the following definition.

\begin{definition}\label[definition]{def:lmv-rmv}
	Let~$G=(A,B,E)$ be a chain graph and let~$M$ be a matching in~$G$.
	Furthermore let~$a \in A$, $b \in B$ with~$\{a,b\} \in M$.
	Then~$\lmv(b,M)$ (resp.~$\rmv(a,M)$) is the number of neighbors of~$b$ (resp.~of~$a$) that are to the left of~$a$ (resp. to the right of~$b$); formally:
	\begin{align*}
		\lmv(b,M) := |\{a' \in N(b) \mid a' \prec a\}|, && \rmv(a,M) := |\{b' \in N(a) \mid b' \prec b\}|.
	\end{align*}
\end{definition}

In~\cref{def:lmv-rmv} the terms ``left'' and ``right'' refer to the ordering of the vertices of $A$ and $B$ 
in the bipartite representation of $G$, as illustrated in \cref{fig:chain-graph-augmenting-path-difficult}. 
The abbreviation $\rmv$ ($\lmv$) stands for ``number of vertices \textbf{r}ight (\textbf{l}eft) of the \textbf{m}atched \textbf{v}ertex''.
We set $$\lmv(a,M) := \lmv(b,M)$$ and $$\rmv(b,M) := \rmv(a,M).$$
Finally, we define~$$\rmv(a_1,a_2,M) := \min_{a_1 \prec a' \prec a_2}\{\rmv(a',M)\}$$ for~$a_1, a_2 \in A$ and~$$\lmv(b_1,b_2,M) := \min_{b_2 \prec b' \prec b_1}\{\lmv(b',M)\}$$ for~$b_1, b_2 \in B$.
For example, in the graph displayed in \cref{fig:chain-graph-augmenting-path-difficult}, we have~$\lmv(b_2,b_3,M) = 2$ and~$\lmv(b_3,b_6,M) = 1$.

With these definitions, we can show a limit on the number of augmenting paths that can ``cross'' an edge in~$M_{G-X}$.

\begin{lemma}\label[lemma]{lem:min-lmv-rmv-upper-bounds-augm-paths}
	Let~$G=(A,B,E)$ be a chain graph and~$M$ be a maximum matching for~$G$ computed by \cref{alg:chain-graphs-linear}.
	Let~$a, b \in V$ with~$\{a,b\} \in M$.
	Then the number of vertex-disjoint alternating paths that (1) start and end with edges not in~$M$ and that (2) have endpoints left of~$a$ and right of~$b$ is at most~$\min\{\lmv(b,M),\rmv(a,M)\}$.
\end{lemma}

\begin{proof}
	We prove the case~$\lmv(b,M) \le \rmv(a,M)$, that is, $\min \{\lmv(b,M), \allowbreak \rmv(a,M)\} = \lmv(b,M)$.
	The case~$\lmv(b,M) > \rmv(a,M)$ follows by symmetry (with switched roles of~$a$ and~$b$).
	Let~$\noAug$ denote the number of vertex-disjoint alternating paths from~$\{a' \in A \mid a' \prec a \}$ to~$\{b' \in B \mid b' \prec b\}$ such that the first and last edge are not in~$M$ (see \cref{fig:chain-graph-augmenting-path-difficult} for an example with~$\noAug = 1$ for~$a = a_4$ and~$b=b_4$).
	Furthermore, let~$a^b_1, \ldots, a^b_{\lmv(b,M)}$ be the neighbors of~$b$ that are to the left of~$a$, that is, $a^b_1 \prec a^b_2 \prec \ldots \prec a^b_{\lmv(b,M)} \prec a$. 
	Since~$G$ is a chain graph it follows that no vertex~$a' \in A$ with~$a' \prec a^b_1$ is adjacent to any vertex~$b' \in B$ with~$b' \preceq b$. 
	Furthermore, for any edge~$\{a',b'\} \in E$ with~$a \prec a'$ and~$b \prec b'$ it follows from the construction of~$M$ (see \cref{alg:chain-graphs-linear}) that~$\{a',b'\} \notin M$. 
	Hence, any of these alternating paths has to contain at least one vertex from~$a^b_1, \ldots, a^b_{\lmv(b,M)}$. 
	Since the alternating paths are vertex-disjoint it follows that~$\noAug \le \lmv(b,M)$. \qed
\end{proof}

From the previous lemma, we directly obtain the following.

\begin{lemma}\label[lemma]{lem:min-lmv-upper-bounds-augm-paths}
	Let~$G=(A,B,E)$ be a chain graph and let~$M$ be the maximum matching for~$G$ computed by \cref{alg:chain-graphs-linear}.
	Let~$a_1, a_2 \in A$ and $b_1, b_2 \in B$ with~$\{a_1,b_1\},\{a_2,b_2\} \in M$ with~$a_1 \prec a_2$.
	Then there are at most~$\lmv(b_1,b_2,M)$ vertex-disjoint alternating paths that (1) start and end with edges not in~$M$ and that (2) have endpoints left of~$a_1$ and right of~$b_2$.
\end{lemma}

\cref{lem:min-lmv-upper-bounds-augm-paths} states that the number of augmenting paths passing through the ``area'' between~$a_1$ and~$a_2$ is bounded. 
Using this, we want to replace this area by a gadget with~$O(\distPara)$ vertices.
To this end, we need further notation.
For each kept vertex~$v \in K$, we may also keep some bounded number of vertices to the right and to the left of~$v$.
We call these vertices the \emph{left buffer} (\emph{right buffer}) of~$v$.

\begin{definition}\label[definition]{def:buffer}
	Let~$G=(A,B,E)$ be a chain graph and let~$M$ be the maximum matching for~$G$ computed by \cref{alg:chain-graphs-linear}.
	Let~$a_1, a_2 \in A$ and $b_1, b_2 \in B$ with~$\{a_1,b_1\}, \{a_2,b_2\} \in M$ and~$a_1 \prec a_2$.
	Then the (at most)~$\lmv(b_1,b_2,M)$ vertices to the right of~$a_1$ form the \emph{right buffer}~$B^r(a_1,M)$ of~$a_1$; formally, 
	\begin{align*}
		B^r(a_1,M) := \{a \in A \mid {}& a_1 \prec a \wedge  \\ & |\{a' \in A \mid a_1 \prec a' \prec a\}| \le \min\{\lmv(b_1,b_2,M),\distPara\}\}.
	\end{align*}
	Analogously,
	\begin{align*}
		B^\ell(a_2,M) 	:= \{a \in A \mid {}& a \prec a_2 \wedge \\ & |\{a' \in A \mid a \prec a' \prec a_2\}| \le \min\{\lmv(b_1,b_2,M),\distPara\}\}, \\
		B^r(b_1,M) 	:= \{b \in B \mid {}& b \prec b_1 \wedge \\ & |\{b' \in B \mid b \prec b' \prec b_1\}| \le \min\{\lmv(b_1,b_2,M),\distPara\}\}, \\
		B^\ell(b_2,M) 	:= \{b \in B \mid {}& b_2 \prec b \wedge \\ & |\{b' \in B \mid b_2 \prec b' \prec b\}| \le \min\{\lmv(b_1,b_2,M),\distPara\}\}.
	\end{align*}
\end{definition}

Note that in~\cref{def:buffer} each of the sets $B^r(a_1,M)$, $B^\ell(a_2,M)$, $B^r(b_1,M)$, and $B^\ell(b_2,M)$ depends on all four vertices $a_1,a_2,b_1,b_2$; we omit these dependencies from the names for the sake of brevity.

The basic idea is now to delete vertices ``outside'' these buffers.
See \cref{fig:chain-graph-buffer-replacement} for an illustrating example of the following data reduction rule formalizing this idea.
\begin{figure}[t!]
	\centering
	\begin{tikzpicture}[scale=1]
		\tikzstyle{knoten}=[circle,draw,minimum size=17pt,inner sep=2pt,fill=white]

		\def\n{10}
		\def\xScale{1}
		
		\foreach \i in {1,...,\n} {
			\node[knoten] (a\i) at (\xScale * \i - \xScale,2.5) {};
			\node[knoten] (b\i) at (\xScale * \i + \xScale,0) {};
		}

		\foreach \x / \name in {3/{$a_1$},10/{$a_2$},4/{$a^r_1$},5/{$a^r_2$},8/{$a^\ell_1$},9/{$a^\ell_2$}} {
			\node at (\xScale * \x - \xScale,2.5) {\name};
		}
		\foreach \x / \name in {1/{$b_1$},8/{$b_2$},2/{$b^r_1$},3/{$b^r_2$},6/{$b^\ell_1$},7/{$b^\ell_2$}} {
			\node at (\xScale * \x + \xScale,0) {\name};
		}

		\foreach[count=\nr] \x in {1,2,3,5,5,6,7,8,9,10}
			\foreach \y in {1, ..., \x} {
				\path (a\nr) edge[-,draw=black!50] (b\y);
			}
		
		\foreach \x in {3, ..., \n}
		{
			\pgfmathtruncatemacro{\y}{\x - 2};
			\path (a\x) edge[-,very thick] (b\y);
		}
		\begin{pgfonlayer}{background}
			\draw[edge] (a1.center) -- (b1.center) %
			\foreach \i [remember=\i as \xi (initially 1)] in {3,5,7,9} {%
				-- (b\xi.center) -- (a\i.center) -- (b\i.center)%
			}
			-- (b9.center);
			\draw[edge] (a2.center) -- (b2.center) %
			\foreach \i [remember=\i as \xi (initially 2)] in {4,6,8,10} {%
				-- (b\xi.center) -- (a\i.center) -- (b\i.center)%
			}
			-- (b10.center);

			\draw[rounded corners,dashed] ($(a4.north west)+(-0.25,0.25)$) rectangle ($(a5.south east)+(0.25,-0.25)$);
			\draw[rounded corners,dashed] ($(a6.north west)+(-0.25,0.25)$) rectangle ($(a7.south east)+(0.25,-0.25)$);
			\draw[rounded corners,dashed] ($(a8.north west)+(-0.25,0.25)$) rectangle ($(a9.south east)+(0.25,-0.25)$);
			\node at ($(a4.north east)+(0.4,0.4)$) {$B^r(a_1,M_{G-X})$};
			\node at ($(a6.north east)+(0.4,0.4)$) {$A'_D$};
			\node at ($(a8.north east)+(0.4,0.4)$) {$B^\ell(a_{2},M_{G-X})$};
			
			\draw[rounded corners,dashed] ($(b2.north west)+(-0.25,0.25)$) rectangle ($(b3.south east)+(0.25,-0.25)$);
			\node at ($(b2.south east)+(0.4,-0.5)$) {$B^r(b_1,M_{G-X})$};

			\draw[rounded corners,dashed] ($(b6.north west)+(-0.25,0.25)$) rectangle ($(b7.south east)+(0.25,-0.25)$);
			\node at ($(b6.south east)+(0.4,-0.5)$) {$B^\ell(b_2,M_{G-X})$};
			\draw[rounded corners,dashed] ($(a4.north west)+(-0.4,0.8)$) rectangle ($(a9.south east)+(0.4,-0.4)$);
			\node at ($(a6.north east)+(0.4,1)$) {$A'$};
		\end{pgfonlayer}
	\end{tikzpicture}
	
	\bigskip
	
	\begin{tikzpicture}[scale=1]
		\tikzstyle{knoten}=[circle,draw,minimum size=17pt,inner sep=2pt,fill=white]

		\def\n{8}
		\def\xScale{1}
		
		\foreach \i in {1,...,5} {
			\node[knoten] (a\i) at (\xScale * \i - \xScale,2.5) {};
		}
		\foreach \i in {1,...,3} {
			\node[knoten] (b\i) at (\xScale * \i + \xScale,0) {};
		}

		\foreach \i in {6,...,\n} {
			\pgfmathtruncatemacro{\nr}{\i + 2}
			\node[knoten] (a\i) at (\xScale * \nr - \xScale,2.5) {};
		}
		\foreach \i in {4,...,\n} {
			\pgfmathtruncatemacro{\nr}{\i + 2}
			\node[knoten] (b\i) at (\xScale * \nr + \xScale,0) {};
		}

		\foreach \x / \name in {3/{$a_1$},10/{$a_2$},4/{$a^r_1$},5/{$a^r_2$},8/{$a^\ell_1$},9/{$a^\ell_2$}} {
			\node at (\xScale * \x - \xScale,2.5) {\name};
		}
		\foreach \x / \name in {1/{$b_1$},8/{$b_2$},2/{$b^r_1$},3/{$b^r_2$},6/{$b^\ell_1$},7/{$b^\ell_2$}} {
			\node at (\xScale * \x + \xScale,0) {\name};
		}

		\foreach[count=\nr] \x in {1,2,3,4,5,6,7,8}
			\foreach \y in {1, ..., \x} {
				\path (a\nr) edge[-,draw=black!50] (b\y);
			}

		\foreach \x in {3, ..., 8}
		{
			\pgfmathtruncatemacro{\y}{\x - 2};
			\path (a\x) edge[-,very thick] (b\y);
		}
		\begin{pgfonlayer}{background}
			\draw[edge] (a1.center) -- (b1.center) %
			\foreach \i [remember=\i as \xi (initially 1)] in {3,5,7} {%
				-- (b\xi.center) -- (a\i.center) -- (b\i.center)%
			}
			-- (b7.center);
			\draw[edge] (a2.center) -- (b2.center) %
			\foreach \i [remember=\i as \xi (initially 2)] in {4,6,8} {%
				-- (b\xi.center) -- (a\i.center) -- (b\i.center)%
			}
			-- (b8.center);
		\end{pgfonlayer}
	\end{tikzpicture}
	\caption{
		An example for the application of \cref{rule:BufferSizeLimit}.
		Top: A part of a chain graph~$G-X$ is displayed and the thick edges indicate parts of the maximum matching~$M_{G-X}$.
		There are two vertex-disjoint augmenting paths from the two vertices left of~$a_1$ to the two vertices right of~$b_2$ highlighted by gray background.
		Moreover, $\lmv(b_1,b_2,M_{G-X}) = 2$.
		Bottom: The part of the graph after applying \cref{rule:BufferSizeLimit} and the corresponding two vertex-disjoint augmenting paths.
		Note that the only edges that are in the graph below but not above are the edges from vertices in~$B^r(a_1,M_{G-X})$ to~$B^\ell(b_2,M_{G-X})$.
	}
	\label{fig:chain-graph-buffer-replacement}
\end{figure}

\begin{rrule}\label{rule:BufferSizeLimit}
	Let~$(G,\solSize)$ be an instance reduced with respect to \cref{rule:MatchingSizeBounds}.
	Let~$a_1, a_2 \in K \cap A$ with $a_1 \prec a_2$ and~$\{a_1,b_1\},\{a_2,b_2\} \in M_{G-X}$ such that~$A' := \{a \in A \mid a_1 \prec a \prec a_2\}$ is of size at least~$2\cdot\min\{\lmv(b_1,b_2,M),\distPara\} + 1$ and~$A' \cap K = \emptyset$.
	Then delete all vertices in~$A'_D := A' \setminus (B^r(a_1,M_{G-X}) \cup B^\ell(a_2,M_{G-X}))$ and their matched neighbors in~$B$, add all edges between the vertices in the right buffer of~$a_1$ and the vertices in the left buffer of~$b_2$, and decrease~$\solSize$ by~$|A'_D|$.
\end{rrule}

\begin{lemma}\label[lemma]{lem:rule-BufferSizeLimit}
	\cref{rule:BufferSizeLimit} is correct and can be exhaustively applied in~$O(n+m)$ time. 
\end{lemma}

\begin{proof}
	We first introduce some notation and provide some general observations.
	Then we show the correctness in two separate claims.
	Finally, we discuss the running time.
	
	Let~$a_1$, $a_2$, $b_1$, and~$b_2$ be as stated in~\cref{rule:BufferSizeLimit}.
	Denote by~$A'$ (resp.~$B'$) the set of vertices between~$a_1$ and~$a_2$ (resp.~between~$b_1$ and~$b_2$).
	Further denote by~$A'_D \subseteq A'$ and~$B'_D \subseteq B'$ the sets of deleted vertices.
	Note that~$|A'| = |B'|$ and~$|A'_D| = |B'_D|$ since~$M_{G-X}$ was produced by \cref{alg:chain-graphs-linear}.
	Denote the vertices in the buffers of~$a_1, a_2, b_1$, and~$b_2$ by~$$B^x(y_z,M_{G-X}) := \{y^x_1, \ldots, y^x_{\min\{\lmv(b_1,b_2,M_{G-X}),\distPara\}}\}$$ for~$x \in \{r,\ell\}, y \in \{a,b\}, z \in [2]$, and~$x = r \iff z = 1$ (see \cref{fig:chain-graph-buffer-replacement} for examples of the concrete variable identifier).

	Since the input instance is reduced with respect to \cref{rule:MatchingSizeBounds}, it follows that~$\solSize-\distPara \le |M_{G-X}| < \solSize$.
	Denote by~$M_{G'-X} := M_{G-X} \cap E'$ the matching obtained from~$M_{G-X}$ by deleting all edges not in the reduced graph~$G'$.
	Recall that~$\solSize' = \solSize - |A'_D|$. %
	We next show in~\cref{claim:chain-1,claim:chain-2} that the input instance~$(G=(V,E),\solSize)$ is a yes-instance if and only if the produced instance~$(G'=(V',E'),\solSize')$ is a yes-instance. 
	Before we present these two claims, observe that there is a perfect matching between the vertices in~$A'_D$ and~$B'_D$, and thus
	\begin{align}
		\solSize - |M_{G-X}| = \solSize' - |M_{G'-X}|. \label{eq:matching-size-diff}
	\end{align}
	
	\begin{myclaim}\label[myclaim]{claim:chain-1}
		If~$(G,\solSize)$ is a yes-instance, then~$(G',\solSize')$ is a yes-instance.
	\end{myclaim}

	\begin{proofofclaim}
		Recall that~$M^{\max}_G(M_{G-X})$ is a maximum matching for~$G$ minimizing the size of~$M_{G-X} \bigtriangleup M^{\max}_G(M_{G-X})$.
		Since~$(G,\solSize)$ is a yes-instance it holds that~$|M^{\max}_G(M_{G-X})| \ge \solSize$.
		For brevity we set~$ G^M := G(M_{G-X},M^{\max}_G(M_{G-X})) = (V, M_{G-X} \bigtriangleup M^{\max}_G(M_{G-X}))$. 
		Note that $G^M$ is a graph that only contains odd-length paths (see \cref{sec:prelim}). 
		We will show that there are as many vertex-disjoint augmenting paths for~$M_{G'-X}$ in~$G'$ as there are paths in~$G^M$. 
		This will show that~$G'$ contains a matching of size
		\begin{align*}
			|M_{G'-X}| + |M^{\max}_G(M_{G-X})| - |M_{G-X}| 	& \ge |M_{G'-X}| + \solSize - |M_{G-X}| \\
															& \overset{\eqref{eq:matching-size-diff}}{=} |M_{G'-X}| + \solSize' - |M_{G'-X}| = \solSize'.		 
		\end{align*}
		To this end, observe that all paths that do not use vertices in~$V \setminus V' = A'_D\cup B'_D$ are also contained in~$G'$.
		Thus, consider the paths in~$G^M$ that use vertices in~$V \setminus V'$.
		Denote by~$\mathcal{P}^M$ the set of all paths in~$G^M$ using vertices in~$V \setminus V'$ and set~$t := |\mathcal{P}^M|$.
		Consider now an arbitrary~$i \in [t]$, and let~$P^M_i \in \mathcal{P}^M$. 
		Denote by~$v^i_1, v^i_2, \ldots, v^i_{p_i}$ the vertices in~$P^M_i$ in the corresponding order, that is, $v^i_1$ and~$v^i_{p_{i}}$ are the endpoints of~$P^M_i$ and we have~$\{v^i_{2j},v^i_{2j+1}\} \in M_{G-X}$ for all~$j \in [p_i/2]$.
		Observe that exactly one endpoint of~$P^M_i$ is in~$A$ and the other endpoint is in~$B$, since $P^M_i$ is an odd-length path.
		Assume without loss of generality that~$v^i_1 \in A$ and~$v^i_{p_i} \in B$.
		Thus, the vertices in~$P^M_i$ with odd (even) index are in~$A$ ($B$).

		We next show that for any two vertices~$v^i_j, v^i_\ell$ of~$P^M_i$ with~$j < \ell < p_i - 1$ and both being in~$A \setminus X$, it follows that~$v^i_j \prec v^i_\ell$.
		First, observe that if~$j = 1$, then~$v^i_j \in A \setminus X$ is a free vertex wrt.~$M_{G-X}$.
		Since~$v^i_\ell$ is matched wrt.~$M_{G-X}$ and since~$M_{G-X}$ is computed by \cref{alg:chain-graphs-linear}, it follows that~$v^i_1 \prec v^i_\ell$.
		Thus, assume that~$j > 1$ and~$\ell > 1$ (thus~$j \ge 3$ and~$j > 1$).
		Assume toward a contradiction that~$v^i_\ell \prec v^i_j$.
		Since~$\ell < p_i - 1$, we have~$v^i_{\ell + 1} \in B \setminus X$ and since~$G-X$ is a chain graph, it follows that~$\{v^i_j,v^i_{\ell+1}\} \in E$, a contradiction to \cref{obs:shortest-augmenting-paths}.
		Thus, $v^i_j \prec v^i_\ell$.
		
		We next show that the path~$P^M_i$ contains at least one vertex~$v^i_j$ left of~$a_1$ and at least one vertex~$v^i_\ell$ right of~$b_2$.
		Recall that~$M_{G-X}$ was computed by \cref{alg:chain-graphs-linear} and, thus, the free vertices are the smallest wrt.\ the ordering~$\prec$ (see also \cref{fig:chain-graph}).
		Thus, if one endpoint of~$P^M_i$ is in~$(A \cup B) \setminus X$, then this vertex is either~$v^i_1$ and left of~$a_1$ or it is~$v^i_{p_1}$ and right of~$b_2$.
		Thus, assume that the endpoints of~$P^M_i$ are in~$X$.
		We showed in the previous paragraph that~$v^i_3 \prec v^i_5 \prec \ldots \prec v^i_{p_1 - 3}$.
		Thus, we also have~$v^i_{p_i-2} \prec v^i_{p_i-4} \prec \ldots \prec v^i_2$ since~$M_{G-X}$ is computed by \cref{alg:chain-graphs-linear}.
		Since we assumed some vertices of~$P^M_i$ to be in~$V \setminus V'$, it follows that for at least one vertex~$v^i_{2j+1}$ it holds that~$a_1 \prec v^i_{2j+1} \prec a_2$.
		Furthermore since by assumption no vertex between~$a_1$ and~$a_2$ or between~$b_1$ and~$b_2$ is in~$K$, it follows that~$v^i_3 \prec a_1$ (since~$v^i_3 \in K$) and~$v^i_{p_i-2} \prec b_2$ (since~$v^i_{p_i-2} \in K$).

		For each~$i \in [t]$ denote by~$a^{P^M_i}$ the last vertex on the path~$P^M_i$ that is not right of~$a_1$, that is, $a^{P^M_i}$ is the vertex on~$P^M_i$ such that for each vertex~$a' \in A \setminus X$ that is in~$P^M_i$ it holds that~$a_1 \prec a'$ or~$a' \preceq a^{P^M_i}$. 
		It follows from the previous paragraph that~$a^{P^M_i}$ exists.
		Analogously to~$a^{P^M_i}$, for each~$i \in [t]$ denote by~$b^{P^M_i}$ the first vertex on the path~$P^M_i$ that is not left of~$b_2$, that is, $b^{P^M_i}$ is the vertex on~$P^M_i$ such that for each vertex~$b' \in B$ that is in~$P^M_i$ it holds that~$b_2 \prec b'$ or~$b' \preceq b^{P^M_i}$.
		This means that in~$G^M$ there is for each~$i \in [t]$ an alternating path from~$a^{P^M_i}$ to~$b^{P^M_i}$
		starting and ending with non-matched edges and all these paths are pairwise vertex-disjoint.
		We show that also in~$G'$ there are pairwise vertex-disjoint alternating paths from~$a^{P^M_i}$ to~$b^{P^M_i}$.
		Assume without loss of generality that~$a^{P^M_1} \prec a^{P^M_2} \prec \ldots \prec a^{P^M_t}$.
		Since in each path~$P^M_i$, $i \in [t]$, the successor of~$a^{P^M_i}$ is to the right of~$b_1$, it follows that 
		$a^{P^M_i}$ has at least~$i$ neighbors right of~$b_1$.
		Since the right buffer of~$b_1$ contains the~$\lmv(b_1,b_2,M_{G-X}) \ge t$ (see \cref{lem:min-lmv-upper-bounds-augm-paths}) vertices to the right of~$b_1$, we have~$\{a^{P^M_i},b^r_i\} \in E$.
		By symmetry, we have~$\{b^{P^M_i},a^\ell_i\} \in E$.
		Recall that~$M_{G-X}$ forms a perfect matching between~$B^r(b_1,M_{G-X})$ and~$B^r(a_1,M_{G-X})$ as well as between~$B^\ell(a_2,M_{G-X})$ and~$B^\ell(b_2,M_{G-X})$.
		Since \cref{rule:BufferSizeLimit} added all edges between~$B^r(a_1,M_{G-X})$ and~$B^\ell(b_2,M_{G-X})$ to~$E'$, it follows that each path~$P^M_i$ can be completed as follows: $a^{P^M_i},b^r_i,a^r_i,b^\ell_i,a^\ell_i,b^{P^M_i}$; note that exactly the edges~$\{b^r_i,a^r_i\}$ and~$\{b^\ell_i,a^\ell_i\}$ are in~$M_{G-X}$.
		Thus, each path in~$\mathcal{P}^M$ can be replaced by an augmenting path for~$M_{G'-X}$ in~$G'$ and all these augmenting paths are vertex-disjoint.
		Thus, there are as many augmenting paths for~$M_{G'-x}$ in~$G'$ as there are paths in~$G^M$ and therefore $(G',\solSize')$ is a yes-instance.
	\end{proofofclaim}

	\begin{myclaim}\label[myclaim]{claim:chain-2}
		If~$(G',\solSize')$ is a yes-instance, then~$(G,\solSize)$ is a yes-instance.
	\end{myclaim}
	\begin{proofofclaim}
		Let~$M_{G'}$ a maximum matching for~$G'$.
		Observe that~$|M_{G'}| \ge \solSize'$.
		We construct a matching~$M_G$ for~$G$ as follows.
		First, copy all edges from~$M_{G'} \cap E$ into~$M_G$.
		Second, add all edges from~$M_{G-X} \cap \binom{V \setminus V'}{2}$, that is, a perfect matching between~$A'_D$ and~$B'_D$ is added to~$M_G$.
		Observe that if all edges in~$M_{G'}$ are also in~$E$, then~$M_G$ is a matching of size~$\solSize$ in~$G$.
		Thus, assume that some edges in~$M_{G'}$ are not in~$E$, that is, $\{a^r_{i_1},b^\ell_{j_1}\}, \ldots, \{a^r_{i_t},b^\ell_{j_t}\} \in M_{G'}\setminus E$ for some~$t \in [\lmv(b_1,b_2,M_{G-X})]$.
		Observe that~$\solSize - |M_G| \le t$.
		Clearly, the vertices~$a^r_{i_1}, \ldots, a^r_{i_t},b^r_{j_1}, \ldots, b^r_{j_t}$ are free with respect to~$M_G$.

		We show that there are $t$~pairwise vertex-disjoint augmenting paths from~$\{a^r_{i_{1}},\ldots,a^r_{i_{t}}\}$ to~$\{b^\ell_{j_{1}},\ldots,b^\ell_{j_{t}}\}$; note, however, that these paths are not necessarily from~$a^r_{i_r}$ to $b^\ell_{j_r}$, where $r \in [t]$.
		To this end, recall that, by definition of~$\lmv(b_1,b_2,M_{G-X})$, each vertex~$b \in B$ with $b_2 \prec b \prec b_1$ has at least~$\lmv(b_1,b_2,M_{G-X})$ neighbors to the left of its matched neighbor.
		This allows us to iteratively find augmenting paths as follows:
		To create the~$q^\text{th}$ augmenting path~$P_q$ start with some vertex~$b^\ell_{j_{q}}$.
		Denote by~$v$ the last vertex added to~$P_q$ (in the beginning we have~$v = b^\ell_{j_{q}}$).
		If~$v \in A$, then add to~$P_q$ the neighbor matched to~$v$.
		If~$v \in B$, then do the following: if~$v$ is adjacent to a vertex~$a\in\{a^r_{i_{1}},\ldots,a^r_{i_{t}}\}$, then add~$a$ to~$P_q$, otherwise add the leftmost neighbor of~$v$ to~$P_q$.
		Repeat this process until~$P_q$ contains a vertex from~$\{a^r_{i_{1}},\ldots,a^r_{i_{t}}\}$.
		After we found~$P_q$, remove all vertices of~$P_q$ from~$G$.
		If~$q < t$, then continue with~$P_{q+1}$.
		Observe that any two vertices of~$P_q$ that are in~$A$ have at least~$\lmv(b_1,b_2,M_{G-X})-1$ other vertices of~$A$ between them (in the ordering of the vertices of $A$, see Figure~\ref{fig:chain-graph}). 
		Thus, after a finite number of steps, $P_q$ will reach a vertex in~$\{a^r_{i_{1}},\ldots,a^r_{i_{t}}\}$.
		Furthermore, it follows that after removing the vertices of~$P_q$ it holds that~$\lmv(b_1,b_2,M_{G-X})$ is decreased by exactly one: $P_q$~contains for each vertex~$b \in B$ at most one vertex among the~$\lmv(b_1,b_2,M_{G-X})$ neighbors of~$b$ that are directly to the left of its matched neighbor in $M_{G-X}$.
		Thus, in each iteration we have~$\lmv(b_1,b_2,M_{G-X}) > 0$.
		It follows that the above procedure constructs $t$~vertex-disjoint augmenting paths from~$\{a^r_{i_{1}},\ldots,a^r_{i_{t}}\}$ to~$\{b^\ell_{j_{1}},\ldots,b^\ell_{j_{t}}\}$.
		Hence, $G$ contains a matching of size~$\solSize$ and thus~$(G,\solSize)$ is a yes-instance.
	\end{proofofclaim}
	\medskip
	
	The correctness of the data reduction rule follows from the previous two claims.
	It remains to prove the running time.
	To this end, observe that the matching~$M_{G-X}$ is given. 
	Computing all degrees of~$G$ can be done in~$O(n+m)$ time.
	Also~$\lmv(v,M_{G-X})$ can be computed in linear time: For each vertex~$b \in B$ one has to check for each neighbor of $b$ whether it is to the left of~$b$'s matched neighbor and to adjust~$\lmv(b,M_{G-X})$ accordingly.
	Furthermore, computing $\lmv(b_1,b_2,M_{G-X})$ and removing the vertices in~$A'_D$ and~$B'_D$ can be done in~$O(\sum_{b \in B'} \deg(b) + \sum_{a \in A'} \deg(a))$~time.
	Thus, \cref{rule:BufferSizeLimit} can be exhaustively applied in~$O(n+m)$ time.  \qed 
\end{proof}

We next upper-bound the number of free vertices with respect to~$M_{G-X}$.
Let $$A_\text{free} := \{ a \in A \mid a \text{ is free with respect to }M_{G-X}\}$$ and~$$A^\distPara_\text{free} := \{a \in A_\text{free} \mid |\{a' \in A_\text{free} \mid a \preceq a'\}| \le \distPara\},$$ that is, $A^\distPara_\text{free}$ contains the~$k$ rightmost free vertices in~$A \setminus X$.
Observe that all vertices in~$A^\distPara_\text{free}$ are left of~$M_{G-X}$.
Analogously, denote by~$B^\distPara_\text{free}$ the set containing the~$k$ leftmost free vertices in~$B \setminus X$.

\begin{rrule}\label{rule:bound-free-vertices}
	Let~$(G,\solSize)$ be an instance reduced with respect to \cref{rule:MatchingSizeBounds}.
	Then delete all vertices in~$\left(A_\text{free} \setminus (K \cup A^\distPara_\text{free})\right) \cup \left(B_\text{free} \setminus (K \cup B^\distPara_\text{free})\right)$.
\end{rrule}

\begin{lemma}\label[lemma]{lem:rule-bound-free-vertices}
	\cref{rule:bound-free-vertices} is correct and can be applied in~$O(n+m)$ time.
\end{lemma}

\begin{proof}
	The running time is clear. It remains to show the correctness.
	Let~$(G,\solSize)$ be the input instance reduced with respect to \cref{rule:MatchingSizeBounds} and let~$(G',\solSize)$ be the instance produced by \cref{rule:bound-free-vertices}. 
	We show that deleting the vertices in~$A_\text{free} \setminus (K \cup A^\distPara_\text{free})$ yields an equivalent instance.
	It then follows from symmetry that deleting the vertices in~$B_\text{free} \setminus (K \cup B^\distPara_\text{free})$ yields also an equivalent instance.

	We first show that if~$(G,\solSize)$ is a yes-instance, then also the produced instance~$(G',\solSize)$ is a yes-instance.
	Let~$(G,\solSize)$ be a yes-instance and~$M_G$ be a maximum matching for~$G$.
	Clearly, $|M_G| \ge \solSize$.
	Observe that for each removed vertex~$a \in A_\text{free} \setminus (K \cup A^\distPara_\text{free})$ it holds that every vertex~$a' \in A^{k}_\text{free}$ is to the right of~$a$, that is, $a \prec a'$ and thus~$N_{G-X}(a) \subseteq N_{G-X}(a')$.
	Since~$(G,\solSize)$ is reduced with respect to \cref{rule:MatchingSizeBounds}, it follows that~$|M_{G-X}| \ge |M_G| - \distPara$.
	Thus, there exist at most~$\distPara$ augmenting paths for~$M_{G-X}$ in~$G$.
	If none of these augmenting paths ends in a vertex~$a \in A_\text{free} \setminus (K \cup A^\distPara_\text{free})$, then all augmenting paths exist also in~$G'$ and thus~$(G',\solSize)$ is a yes-instance.
	If one of these augmenting paths, say~$P$, ends in~$a$, then at least one vertex~$a' \in A^\distPara_\text{free}$ is not endpoint of any of these augmenting paths.
	Since~$a \notin K$, it follows from the definition of~$K$ 
	that the neighbor~$b$ of~$a$ on~$P$ is indeed in~$B \setminus X$.
	Since~$N_{G-X}(a) \subseteq N_{G-X}(a')$, it follows that~$\{a',b\} \in E$ and thus we can replace~$a$ by~$a'$ in the augmenting path.
	By exhaustively applying the above exchange argument, it follows that we can assume that none of the augmenting paths uses a vertex in~$A_\text{free} \setminus (K \cup A^\distPara_\text{free})$. 
	Thus, all augmenting paths are also contained in~$G'$ and hence the resulting instance~$(G',\solSize)$ is still a yes-instance.
	
	Finally observe that if~$(G',\solSize)$ is a yes-instance, then also~$(G,\solSize)$ is a yes-instance: any matching of size~$\solSize$ in~$G'$ is also a matching in~$G$ since~$G'$ is a subgraph of~$G$.
	\qed
\end{proof}

We have now all statements that we need to show our second main result.

\begin{theorem}
\label{thm:cubic-kernel-chain-graphs}
	\Match on bipartite graphs admits a cubic-vertex kernel with respect to the vertex deletion distance to chain graphs. 
	The kernel can be computed in linear time.
\end{theorem}

\begin{proof}
	Let~$(G,\solSize)$ be the input instance with~$G = (V,E)$, the two partitions $V = A \cup B$, and~$X \subseteq V$ such that~$G-X$ is a chain graph.
	If~$X$ is not given explicitly, then use the linear-time factor-four approximation provided in \cref{lem:linTimeApproxDistChainGraphs} to compute~$X$.
	The kernelization is as follows:
	First, compute the matching~$M_{G-X}$ in linear time with \cref{alg:chain-graphs-linear}.
	Next compute the set of kept vertices~$K$.
	Then, apply \cref{rule:reduce-neighborhood-of-X,rule:MatchingSizeBounds,rule:BufferSizeLimit,rule:bound-free-vertices}.
	By \cref{lem:rule-BufferSizeLimit,lem:rule-bound-free-vertices}, this can be done in linear time.
	Let~$b^K_\ell$ the leftmost vertex in~$K \cap B$ and~$a^K_r$ rightmost vertex in~$A \cap K$.
	Let~$a^K_\ell$ and~$b^K_r$ be their matched neighbors.
	Since we reduced the instance with respect to \cref{rule:reduce-neighborhood-of-X}, we have~$|K| \le 2\distPara^2$.
	Moreover, as we reduced the instance with respect to \cref{rule:BufferSizeLimit}, it follows that the number of vertices between~$a^K_\ell$ and~$a^K_r$ as well as the number of vertices between~$b^K_\ell$ and~$b^K_r$ is at most~$4\distPara^3$, respectively.
	Furthermore, there are at most~$2\distPara$ free vertices left in~$V \setminus X$ since we reduced the instance with respect to \cref{rule:bound-free-vertices}.
	It remains to upper-bound the number of matched vertices left of~$b^K_\ell$ and right of~$a^K_r$ (see \cref{fig:chain-graph-borders}).
	
\begin{figure}
	\centering
	\begin{tikzpicture}[scale=1]
		\tikzstyle{knoten}=[circle,draw,minimum size=17pt,inner sep=2pt,fill=white]

		\def\n{10}
		\def\xScale{1.2}
		
		\foreach \i in {1,...,\n} {
			\node[knoten] (a\i) at (\xScale * \i,2.5) {};
			\node[knoten] (b\i) at (\xScale * \i,0) {};
		}

		\foreach[count=\nr] \x in {1,...,\n}
			\foreach \y in {1, ..., \x} {
				\path (a\nr) edge[-,draw=black!50] (b\y);
			}
		
		\foreach \x in {3, ..., \n}
		{
			\pgfmathtruncatemacro{\y}{\x - 2};
			\path (a\x) edge[-,very thick] (b\y);
		}
		
		\foreach \x / \name in {5/{$a^K_\ell$},8/{$a^K_r$}} {
			\node[] at (\xScale * \x,2.5) {\name};
		}
		\foreach \x / \name in {3/{$b^K_\ell$},6/{$b^K_r$}} {
			\node[] at (\xScale * \x,0) {\name};
		}

		\begin{pgfonlayer}{background}
			\draw[rounded corners,dashed] ($(a1.north west)+(-0.25,0.25)$) rectangle ($(a2.south east)+(0.25,-0.25)$);
			\node at ($(a1.north east)+(0.4,0.4)$) {$\le 2k$ free vertices};
			\draw[rounded corners,dashed] ($(a5.north west)+(-0.25,0.25)$) rectangle ($(a8.south east)+(0.25,-0.25)$);
			\node at ($(a6.north east)+(0.4,0.4)$) {$\le 4k^3$ vertices};
			
			\draw[rounded corners,dashed] ($(b9.north west)+(-0.25,0.25)$) rectangle ($(b10.south east)+(0.25,-0.25)$);
			\node at ($(b9.south east)+(0.4,-0.5)$) {$\le 2k$ free vertices};
			\draw[rounded corners,dashed] ($(b3.north west)+(-0.25,0.25)$) rectangle ($(b6.south east)+(0.25,-0.25)$);
			\node at ($(b4.south east)+(0.4,-0.5)$) {$\le 4k^3$ vertices};
		\end{pgfonlayer}
	\end{tikzpicture}
	\caption{
		Schematic representation of the situation in the proof of \cref{thm:cubic-kernel-chain-graphs}; only the chain graph~$G-X$ is shown.
		The vertices within the dashed boxes are bounded by the applications of \cref{rule:reduce-neighborhood-of-X,rule:bound-free-vertices,rule:MatchingSizeBounds,rule:BufferSizeLimit}.
		Moreover, the vertices~$a^K_r$, $a^K_\ell$, $b^K_r$, and~$b^K_\ell$ are all adjacent to vertices in~$X$.
		It remains to upper-bound the vertices right of~$a^K+r$ and left of~$b^K_\ell$.
	}
	\label{fig:chain-graph-borders}
\end{figure}

	Observe that all vertices left of~$b^K_\ell$ are matched with respect to~$M_{G-X}$.
	If there are more than~$2\distPara$ vertices to the left of~$b^K_\ell$, then do the following:
	Add four vertices  $a_\ell, b_\ell, x^a_\ell, x^b_\ell$ to~$V$.
	The idea is that~$\{a_\ell, b_\ell\}$ should be an edge in~$M_{G-X}$ such that~$a_\ell \in A$ and~$b_\ell \in B$ are in~$K$ and there is no vertex left of~$b_\ell$.
	This means we add these vertices to simulate the situation where the leftmost vertex in~$B \setminus X$ is also in~$K$.
	To ensure that~$a_\ell$ and~$b_\ell$ are in~$K$ and that they are not matched with some vertices in~$G$, we add~$x^a_\ell$ and $x^b_\ell$ to~$X$ and make~$x^a_\ell$ respectively $x^b_\ell$ to their sole neighbors.
	In this way, we ensure that there is maximum matching in the new graph that is exactly two edges larger than the maximum matching in the old graph.
	In this new graph we can then apply \cref{rule:BufferSizeLimit} to reduce the number of vertices between~$b_\ell$ and~$b^K_\ell$.
	Formally, we add the following edges.
	Add~$\{a_\ell,x^a_\ell\}, \{b_\ell,x^b_\ell\}$ to~$E$. 
	Add all edges between~$b_\ell$ and the vertices in~$B \setminus X$.
	Let~$a$ be the rightmost vertex in~$A^\distPara_\text{free}$.
	Then, add edges between~$a_\ell$ and~$N_{G-X}(a)$.
	Set~$b' \prec b_\ell$ for each~$b' \in B\setminus X$, set~$a \prec a_\ell$ for each vertex~$a \in A_\text{free}$, and set~$a_\ell \prec a'$ for each matched vertex in~$A \setminus X$.
	Furthermore, add~$\{a_\ell, b_\ell\}$ to~$M_{G-X}$ and add~$a_\ell$ and~$b_\ell$ to~$K$.
	Finally, increase $\solSize$ by two.
	Next, apply \cref{rule:BufferSizeLimit} in linear time, then remove~$a_\ell, b_\ell, x^a_\ell, x^b_\ell$ and reduce~$\solSize$ by two.
	After this procedure, it follows that there are at most~$2\distPara$ vertices left of~$b^K_\ell$.
	If there are more than~$2\distPara$ vertices right of the rightmost vertex~$a^K_r$ in~$A \cap K$, then use the same procedure as above.
	Thus, the total number of vertices in the remaining graph is at most~$|X| + 2\distPara + 4\distPara^3 = O(\distPara^3)$.
	Furthermore, observe that adding and removing the four vertices as well as applying \cref{rule:BufferSizeLimit} can be done in linear time.
	Thus, the overall running time of the kernelization is~$O(n+m)$. \qed
\end{proof}

Applying an~$O(n^{2.5})$-time algorithm for \BipMatch~\cite{HK73} on the kernel yields the following.

\begin{corollary}\label[corollary]{cor:chain-graph}
	\BipMatch can be solved in~$O(k^{7.5} + n + m)$ time, where~$k$ is the vertex deletion distance to chain graphs.
\end{corollary}

Using the randomized~$O(n^{\omega})$-time algorithm for \Match~\cite{MS04}, one would obtain a randomized algorithm with running time~$O(k^{3\omega} + n + m)$.
Here, $\omega < 2.373$ is the matrix multiplication coefficient, that is, two $n \times n$ matrices can be multiplied in~$O(n^\omega)$ time.

\section{Conclusion}
We focussed on kernelization results for unweighted (\textsc{Bipartite}) \Match.
There remain numerous challenges for future research as discussed in the second part of this concluding section. 
First, however, let us discuss the closely connected issue of FPTP~algorithms for \Match{}. 
There is a generic augmenting path-based approach to provide FPTP algorithms for \Match:
Note that one can find an augmenting path in linear time~\cite{Blu90,GT91,MV80}. 
Now the solving algorithm for \Match{} parameterized by some vertex deletion distance~$k$ works as follows:
\begin{enumerate}
	\item Use a constant-factor linear-time (approximation) algorithm to compute a vertex set~$X$ such that~$G-X$ is a ``trivial'' graph (where \Match is linear-time solvable).
	\item Compute in linear time an initial maximum matching~$M$ in~$G-X$.
	\item Start with~$M$ as a matching in~$G$ and increase the size at most~$|X| = \distPara$ times to obtain in~$O(\distPara \cdot (n+m))$ time a maximum matching for~$G$.
\end{enumerate}
From this we can directly derive that
\Match can be solved in~$O(k(n+m))$ time, where~$k$ is one of the following parameters: feedback vertex number, feedback edge number, and vertex cover number. 
Moreover, \BipMatch can be solved in~$O(k(n+m))$ time, where~$k$ is the vertex deletion distance to chain graphs. 
Using our kernelization results, the multiplicative dependence of the running time on parameter~$k$ can now be made an additive one.
For instance, in this way the running time for \textsc{Bipartite Matching} parameterized by vertex deletion distance to chain graphs ``improves'' from $O(k(n+m))$ to $O(k^{7.5}+n+m)$.

We conclude with some questions and tasks for future research.
Can the size or the running time of the kernel with respect to feedback vertex set (see \cref{sec:general-case}) be improved?
In particular, can the exponential upper bound on the kernel size be decreased to a polynomial upper bound?
Is there a linear-time computable kernel for \Match parameterized by the treedepth~$\ell$ (assuming that $\ell$ is given)?
This would complement the recent $O(\ell m)$~time algorithm~\cite{IOO18}.
Can one extend the kernel of \cref{sec:bipartite-case} from \BipMatch to \Match parameterized by the distance to chain graphs?

\paragraph{Acknowledgment.} We thank anonymous reviewers of \emph{Algorithmica} for their valuable feedback. 

\bibliographystyle{abbrvnat}
\bibliography{bib} %

\end{document}